\documentclass{amsart}

\title[On the full Kostant-Toda hierarchy and its $\ell$-banded reductions]
{On the full Kostant-Toda hierarchy and its $\ell$-banded reductions for the Lie algebras of type $A, B$ and $G$}

\author{Yuji Kodama$^1$ and Yuancheng Xie$^2$} 
\date{\today}
\address{
1. College of Mathematics and Systems Sciences, Shandong University of Science and Technology, Qingdao 266590, China, and\\ Department of Mathematics, Ohio State University, 
Columbus, OH 43210, USA
}
\email{kodama@math.ohio-state.edu}
\address{2. Beijing International Center for Mathematical Research, Peking University,
Beijing 100871, China}
\email{xieyuancheng@bicmr.pku.edu.cn}

\usepackage{color}
\usepackage{rotating}

\countdef\x=23
\countdef\y=24
\countdef\z=25
\countdef\t=26

\def\tbox(#1,#2)#3{
\x=#1 \y=#2 
\multiply\x by 12 
\multiply\y by 12 
\z=\x \t=\y
\advance\z by 12 
\advance\t by 12 
\psline(\x,\y)(\x,\t)(\z,\t)(\z,\y)(\x,\y)
\advance\x by 6
\advance\y by 6 
\rput(\x,\y){{\bf #3}}}

\usepackage{amssymb}
\usepackage{graphicx}
\usepackage{epstopdf}
\usepackage{amscd}
\usepackage{appendix}
\usepackage{graphics}
\usepackage{MnSymbol}

\usepackage{geometry}
\usepackage{hyperref}

\def\proof{\par{\it Proof}. \ignorespaces}
\def\endproof{{\ \vbox{\hrule\hbox{%
     \vrule height1.3ex\hskip0.8ex\vrule}\hrule }}\par}

\numberwithin{equation}{section}

\let\trueint=\int
\let\truesum=\sum
\def\int{\mathop{\textstyle\trueint}\limits}
\def\sum{\mathop{\textstyle\truesum}\limits}
\let\<=\langle
\let\>=\rangle

\def\HH{\mathcal{H}}

\def\G{{\mathcal G}}

\def\Sym{{\mathfrak S}}

\def\t{\mathbf{t}}
\def\0{\mathbf{0}}
\def\ve{\varepsilon}

\def\wb{\hskip0.4cm\circle{8}}
\def\bb{\hskip0.4cm\circle*{8}}
\def\height{\text{ht}}
\def\g{\mathfrak{g}}
\def\Int{\mathbb{Z}}
\def\spl{\mathfrak{sl}}

\def\edge{\ar@{-}}
\def\dedge{\ar@{.}}

\usepackage{amsmath, amsthm, amssymb, amsbsy,amsfonts,amscd}
\usepackage[mathscr]{eucal}
\usepackage{epsfig}
\usepackage{young}
\newtheorem{theorem}{Theorem}[section]

\newtheorem{proposition}[theorem]{Proposition}
\newtheorem{lemma}[theorem]{Lemma}

\newtheorem{corollary}[theorem]{Corollary}

\newtheorem{remark}[theorem]{Remark}

\usepackage{amsfonts}
\usepackage{xy}
\usepackage{amssymb}
\usepackage{amsmath}

\newcommand{\Z}{\mathbb Z}

\newcommand{\C}{\mathbb C}

\newcommand{\h}{\mathfrak{h}}
\newcommand{\vep}{\varepsilon}

\DeclareMathOperator{\ad}{ad}

\makeatletter
\renewcommand*\env@matrix[1][*\c@MaxMatrixCols c]{%
  \hskip -\arraycolsep
  \let\@ifnextchar\new@ifnextchar
  \array{#1}}
\makeatother

\newcommand{\thmrefer}[1]{\renewcommand\thetheorem
  {\protect\ref{#1}}\addtocounter{theorem}{-1}}

\xyoption{all}
\CompileMatrices

\begin{document}

\begin{abstract}
This paper concerns the solutions of the full Kostant-Toda (f-KT) hierarchy in the Hessenberg form and their reductions to the $\ell$-banded Kostant-Toda ($\ell$-KT) hierarchy.  We also study the f-KT hierarchy and the corresponding $\ell$-KT hierarchy on simple Lie algebras of type $A, B$ and $G$ based on root space reductions with proper Chevalley systems. Explicit formulas of the polynomial solutions for the $\tau$-functions are also given in terms of the Schur functions and Schur's $Q$-functions.
\end{abstract}

\maketitle

\setcounter{tocdepth}{1}
\tableofcontents


\section{Introduction}
We study reductions and rational solutions of the full Kostant-Toda (f-KT) hierarchy in the Hessenberg form in this paper. 
In the last fifty years, there have been several publications (c.f. \cite{KS08, Tomei2013}) which studied the classical tri-diagonal Toda lattice and the full Kostant-Toda lattice on the semisimple Lie algebra $\mathfrak{g}=\mathfrak{n}_-\oplus\mathfrak{b}_+$ in the standard notation (c.f. \cite{KS08}). These sytems can be viewed as Toda flows on the minimal and maximal symplectic leaves inside the dual space $\mathfrak{b}_+^*$, respectively (c.f. \cite{GS}). It is natural to investigate the Toda flows on all the other symplectic leaves,  however intermediate Toda lattices which sit between the tri-diagonal Toda lattice and the full Kostant-Toda lattice were rarely mentioned in the literature (c.f. \cite{Nanda1982, Damianou-Sabourin-Vanhaecke2015}). One of the main goals in this paper is to give a clean characterization for the special intermediate Toda lattices possessing the banded structures in the Hessenberg form based on a root space decomposition \cite{KY}. The geometry of the Hessenberg matrices however naturally leads to a further motivation for the current project which we now explain. 

Let $\mathsf{X}$ be a linear map $\mathsf{X}: \mathbb{C}^M \to \mathbb{C}^M$ and $h: \{1, 2, \dots, M\} \to \{1, 2, \dots, M\}$ a non-decreasing function satisfying $h(i) \ge i$ for $1 \le i \le M$. The \emph{Hessenberg variety} associated to $X$ and $h$ is defined as
\[\text{Hess}(\mathsf{X}, h) := \{V_{\bullet} \in Fl(\mathbb{C}^M) \ \vert\ \mathsf{X}V_i \subseteq V_{h(i)} \text{ for all } 1 \le i \le M\},\]
where $Fl(\mathbb{C}^M)$ is the complete flag variety of $\mathbb{C}^M$ consisting of sequences $V_{\bullet} = (V_1 \subset V_2 \subset \cdots \subset V_M = \mathbb{C}^M)$ of the linear subspaces of $\mathbb{C}^M$ with dim $V_i = i$ for $1 \le i \le M$. The following two examples of Hessenberg varieties are especially relevant here. Consider the linear map on $\mathbb{C}^M$ associated with the regular nilpotent matrix
\begin{equation}\label{eq:pnil}
\mathsf{N} = \begin{pmatrix}
0 & 1 & & & \\
& 0 & 1 & & \\
& & \ddots & \ddots & \\
& & & 0 & 1\\
& & & & 0
\end{pmatrix}.
\end{equation}
For $h = (M, M, \dots, M)$, Hess$(\mathsf{N}, h)$ is the flag variety $Fl(\mathbb{C}^M)$ itself, and for $h = (2, 3, 4, \dots, M, M)$, Hess$(\mathsf{N}, h)$ is the so-called \emph{Peterson variety}. These two varieties are associated with the maximal and minimal Hessenberg functions respectively (c.f. \cite{AH} and the reference therein for more information on the Hessenberg varieties). 

Putting the stories of f-KT hierarchy and Hessenberg varieties side by side, it is hard to believe that these beautiful things are not related in a deep and essential way. Indeed, another way to define the Peterson variety is as follows (c.f. \cite{B17}). Let $SL_M(\mathbb{C})^{\mathsf{N}}$ be the centralizer of $\mathsf{N}$ in $SL_M(\mathbb{C})$ where $\mathsf{N}$ is the nilpotent matrix defined in \eqref{eq:pnil}, then Peterson variety is the following $M-1$ dimensional subvariety of the flag variety
\[\overline{SL_M(\mathbb{C})^{\mathsf{N}} \cdot \mathfrak{b}_-} \subset SL_M(\mathbb{C}) \slash \mathcal{B}^-,\]
where $\mathcal{B}^-$ is the Borel subgroup of lower triangular matrices and $\mathfrak{b}_-=\text{Lie}(\mathcal{B}^-)$.
It may be easy to recognize from this definition that Peterson variety consists of exactly those points in the flag variety which parameterize rational solutions of the tri-diagonal Toda hierarchy (c.f. \cite{FH} and below). This correspondence can actually be used to study geometry of the Peterson variety. As an example we note that in \cite{FH} (Theorem 2.5) Flaschka and Haine essentially proved in the Toda lattice setting that the Peterson variety intersects with Bruhat cells corresponding to the longest Weyl group elements generated by the reflections with respect to the hyperplanes orthogonal to a subset of simple roots, and this fact was only proved later in the Hessenberg variety setting by Harada and Tymoczko in \cite{HT} (Proposition 5.8).

As we also proved in \cite{Xie:22b} that there is a one-to-one correspondence between rational solutions of the f-KT hierarchy and points in the corresponding flag variety,  it is natural to anticipate a more precise correspondence between sub-hierarchies of the f-KT hierarchy and sub-Hessenberg varieties of the flag variety and to use results from one side to gain insights on the other. Such initial attempts had already been made in \cite{AC, AB}, and our humble goal in this paper is to push the story on the f-KT hierarchy side a little bit by studying solutions of the banded f-KT hierarchy while leaving the investigation of the precise correspondence in a future publication.

The f-KT hierarchy can be defined on any semisimple Lie algebra (c.f. \cite{GS, KO:22, Xie:22b}) and an immediate reason for our restriction to the Hessenberg form Lax matrix (which includes Lie algebras of type $A, B, C, G$ etc.)
\[L = \sum_{i=1}^{M-1}E_{i,i+1}+\sum_{1\le j\le i\le M}a_{i,j}(\t)E_{i,j}\]
is that all the matrix elements $a_{i, j}(\t)$ can be neatly expressed in terms of $\tau$-functions where $\t = (t_1, t_2, \dots)$. Let $\Delta_N(\t)$ denote the $N$-th leading principal minor of the matrix $\exp(\Theta_{L(\0)}(\t))$,
\begin{equation}
\Delta_N(\t):=[\exp(\Theta_{L(\0)}(\t)))]_N\qquad\text{for}\quad N=1,2,\ldots,M,
\end{equation}
where $\Theta_{L(\0)}(\t) := \sum_{n = 1}^{\infty}(L(\0))^nt_n$. Then we have the following proposition (Proposition \ref{prop:aij}):

\begin{theorem}\label{prop:aijintro}
The elements $a_{i,j}(\t)$ for the f-KT hierarchy can be expressed as
\begin{equation}\label{aijintro}
a_{i,j}(\t)=\frac{p_{i-j+1}(\tilde D)\Delta_j(\t)\circ \Delta_{i-1}(\t)}{\Delta_j(\t)\Delta_{i-1}(\t)} \qquad\text{for}\quad 1\le j\le i\le M,
\end{equation}
where $\tilde D$ is defined by
\[
\tilde D:=\left(D_1,\,\frac{1}{2}D_2,\,\frac{1}{3}D_3,\ldots\,\right)\qquad\text{with}\qquad 
D_nf\circ g=\left(\frac{\partial}{\partial t_n}-\frac{\partial}{\partial t'_n}\right)f(\t)g(\t')\Big|_{\t=\t'}.
\]
\end{theorem}
In particular, we have
\begin{equation}\label{aiiintro}
a_{N,N}(\t)=\frac{\partial}{\partial t_1}\ln\frac{\Delta_N(\t)}{\Delta_{N-1}(\t)}, \qquad 1 \le N \le M.
\end{equation}
For f-KT hierarchy defined on semisimple Lie algebras, the matrix elements can actually all be expressed as polynomials in $a_{N, N}$'s and their derivatives, and thus in terms of $\tau$-functions (Proposition \ref{prop:sol_b}), but in general we do not have a clean formula like \eqref{aijintro}.
The formula \eqref{aijintro} was derived in \cite{AvM} for the discrete KP hierarchy in terms of the vertex operator formalism and we give an elementary proof in Appendix A for the f-KT hierarchy based on the iso-spectral deformation method \cite{KY}. 

The $\ell$-banded Kostant-Toda ($\ell$-KT) hierarchy is the f-KT hierarchy constrained by the condition that the subdiagonals of the Lax matrix $L$ lower than the $\ell$-th diagonal are all zero. This condition is consistent with the Toda flows in the sense that if the initial condition $L(\0)$ is $\ell$-banded, then it stays $\ell$-banded for all $\t \ne \0$ under the Toda flows. Note that for the $\ell$-KT hierarchy the elements on the $\ell$-th diagonal can be directly integrated (Proposition \ref{prop:l-band}):
\begin{proposition}
The elements $a_{i,j}$ in the $\ell$-th diagonal of the $\ell$-KT hierarchy are given by
\begin{equation}\label{eq:aijl-band}
a_{k+\ell, k}(\t)=a_{k+\ell,k}^0\frac{\Delta_{k+\ell}(\t)\Delta_{k-1}(\t)}{\Delta_{k+\ell-1}(\t)\Delta_k(\t)}\qquad\text{for}\quad 1\le k\le M-\ell,
\end{equation}
where $a_{k+\ell,k}^0$ is a constant.
\end{proposition}
Comparing \eqref{aijintro} with \eqref{eq:aijl-band}, we obtain the first characterization of the $\ell$-KT hierarchy in terms of the bilinear relations,
\begin{equation}\label{l-bandintro}
p_{\ell+1}(\tilde D)\Delta_k\circ \Delta_{k+\ell-1}(\t)=a_{k+\ell,k}^0\Delta_{k+\ell}(\t)\Delta_{k-1}(\t)\qquad\text{for}\quad 1\le k\le M-\ell.
\end{equation}

When $\ell=1$, we recover the bilinear equation for the original tridiagonal Toda lattice ($1$-KT lattice),
\[
p_2(\tilde D)\Delta_k\circ\Delta_k=\frac{1}{2}D_1^2\Delta_k\circ\Delta_k=\Delta_k\partial_1^2\Delta_k-(\partial_1\Delta_k)^2=a_{k+1,k}^0\Delta_{k+1}\Delta_{k-1}.
\]

Our second characterization of the $\ell$-KT hierarchy is in terms of the singular structure of their solutions, that is which Bruhat cells the complex solutions might hit. Note that when the Toda flows become singular and hit the boundary of a Schubert cell at time $\t_*$, then we have
\begin{equation}\label{eq:nwbintro}
\exp(\Theta_{L(0)}(\t_*))=n_*\dot w_*b_*\qquad\text{for some}\quad w_*\in \frak{S}_M,
\end{equation}
where $n_*\in \mathcal{N}^-$, $b_*\in \mathcal{B}^+$ and $\dot{w}_* \in GL_M$ is a representative of $w_* \in \frak{S}_M$, the symmetric group of order $M$. Set $\t\to \t+\t_*$, we have the following expansion of $\tau$-functions near $\t_*$ (Theorem \ref{thm:Schurexpand}):
\begin{theorem}\label{thm:Schurexpandintro}
The $\tau$-function $\Delta_N(\t; w_*)$ has the  Schur expansion,
\[
\Delta_N(\t; w_*)=S_{\lambda}(\t)+\sum_{\mu\supset\lambda}c_{\mu}S_\mu(\t),
\]
where $S_\lambda(t)$ is the Schur polynomial associated with the Young diagram $\lambda$, which
is determined by $w_*$ as follows: Let $\{i_1,i_2,\ldots,i_N\}$ be defined by
\[
\{i_1,i_2,\ldots,i_N\}=\{w_*(1),w_*(2),\ldots, w_*(N)\},
\]
with $i_1<i_2<\cdots<i_N$. Then the Young diagram $\lambda=(\lambda_1,\lambda_2,\ldots,\lambda_N)$ is given by
\[
\lambda_k=i_k-k\qquad\text{for}\qquad 1\le k\le N.
\]
\end{theorem}

Our second characterization constraints which Bruhat cells solutions of the $\ell$-KT hierarchy may hit, equivalently Theorem \ref{prop:lbandKTintro} in the following shows how the corresponding symplectic leaves for the $\ell$-KT hierarchy may be compactified in the flag variety (Proposition \ref{prop:lbandKT}):
\begin{theorem}\label{prop:lbandKTintro}
Let $w_*$ be an element in $\frak{S}_M$ satisfying the following conditions,
\begin{itemize}
\item[(a)]
$w_*(i)\ne w_*(i+k)+1\quad\text{for any } k > \ell,$ and
\item[(b)] for some $1\le i\le M-\ell$, there is a relation,
\[
w_*(i)=w_*(i+\ell)+1.
\]
\end{itemize}
Then the corresponding $\tau$-functions give a solution to the $\ell$-KT hierarchy, that is,
for each $m$ with $\ell+1\le m\le M-1$, the elements $a_{i,j}$ in the Lax matrix $L$ satisfy
\[
a_{m+k,k}(\t)=0\qquad\text{for}\qquad 1\le k\le M-m,
\]
and $a_{\ell+k,k}(\t)\ne 0$ for  some $1\le k\le M-\ell$.
\end{theorem}
When $\ell = 1$, we recover the classical result of \cite{FH} in type $A$ of the Lie algebra $\mathfrak{g}$, we mentioned before for the tri-diagonal KT lattice.

We then apply the above machinery to  f-KT hierarchy defined on simple Lie algebras of type $A, B, G$, and for that we give explicitly the Chevalley basis and root space decompositions so that we can put the corresponding Lax matrices into the Hessenberg form. We also embed their Weyl groups into the symmetric group $\mathfrak{S}_M$ so that the second characterization of the $\ell$-KT hierarchy can be applied. When $L(\0)$ is nilpotent, the $\tau$-functions become polynomials of $\t$. We construct all polynomial solutions for f-KT hierarchy in type $A, B$ and $G$. 

It is worth mentioning that the last $\tau$-function in type $B$ possesses a Pfaffian structure (Theorem \ref{thm:pfaffian}, Proposition \ref{duality}, Proposition \ref{prop:Frobenius}, Proposition \ref{prop:SchurQ}):
\begin{theorem}\label{thm:pfaffianintro}
The $\tau$-function associated with $ w_*$ denoted by $\tau^B_n(\t_B; w_*)$ is given by
\[
\tau^B_n(\t_B, w_*)=\sqrt{\Delta_n(\t_B)}=\sum_{r\in\mathfrak{N}_B}\phi_r(\t_B)\eta_r(\xi; w_*),
\]
where $\mathfrak{N}_B \cong \mathbb{Z}_2^n$ is the normal subgroup of the Weyl group $\mathfrak{W}_B = \mathfrak{S}_n \ltimes \mathbb{Z}_2^n$ in type $B$.
The coefficients $(\eta_r(\xi; w_*))_{r \in \mathfrak{N}_B}$ form a spinor and the $\phi_r(\t_B)$'s are exponential functions. For rational solutions, the $\phi_r(\t_B)$'s are given by $Q$-Schur polynomials. 
\end{theorem}
As the Lie algebra of $G_2$-type could be embedded into a Lie algebra of $B_3$-type, the Pfaffian structure pertains in the $\tau$-functions of $G_2$. 

The rest of the paper is organized as follows. In Section \ref{sec:fKT} we introduce our main object f-KT hierarchy in the Hessenberg form and their $\tau$-functions. In Section \ref{sec:Schurexpand} we give Schur expansion of the $\tau$-function $\Delta_N(\t)$. In Section \ref{sec:lbanded}, we give the two characterizations for the $\ell$-KT hierarchy in the Hessenberg form. In Section \ref{sec:Lie} we define f-KT hierarchy and its $\tau$-functions on simple Lie algebras and present some of their basic properties. In Sections \ref{sec:AKT},  \ref{S:BKT} and  \ref{sec:GKT}, we choose proper Chevalley basis to put the f-KT hierarchy in type $A, B, G$ in the Hessenberg form and study their solutions especially the rational solutions and their banded structures.

\section{The full Kostant-Toda lattice}\label{sec:fKT}
The full Kostant-Toda (f-KT) lattice was first introduced in \cite{EFS}, and it takes the following form,
\begin{equation}\label{Lax}
\frac{dL}{dt}=[(L)_{\ge 0},L],
\end{equation}
where the Lax matrix $L$ is the $M\times M$ \emph{Hessenberg matrix}, 
\begin{align}\label{L}
L&=\sum_{i=1}^{M-1}E_{i,i+1}+\sum_{1\le j\le i\le M}a_{i,j}E_{i,j}\\[1.0ex]
&=
\begin{pmatrix}
a_{1,1} & 1 & 0 & \cdots & 0 \\
a_{2,1}&a_{2,2}&1&\cdots &0\\
\vdots &\vdots &\ddots &\ddots & \vdots\\
a_{M-1,1}&a_{M-1,2}&\cdots&\cdots&1\\
a_{M,1}& a_{M,2}&\cdots&\cdots & a_{M,M}
\end{pmatrix},\nonumber
\end{align}
and $(L)_{\geq 0}$ denotes the weakly upper triangular part of $L$. The $E_{i,j}$ is the standard basis of the set of $M\times M$ matrices defined by
\[
E_{i,j}=(\delta_{i,k}\delta_{j,l})_{1\le k,l\le M}.
\]
In terms of the entries $a_{i,j}=a_{i,j}(t)$,  
the f-KT lattice is defined by the system of equations
\begin{equation}\label{eq:f-KTcomp}
\frac{da_{\ell+k,k}}{dt}=a_{\ell+k+1,k}-a_{\ell+k,k-1}+(a_{\ell+k,\ell+k}-a_{k,k})a_{\ell+k,k}
\end{equation}
for  $k=1,\ldots,M-\ell$ and $\ell=0,1,\ldots,M-1$. Here we use the convention
that $a_{i,j}=0$  if 
$j= 0$ or $i= M+1$.  Note  that  the index $\ell$ represents the $\ell$-th subdiagonal of the matrix $L$, that is, $\ell=0$ corresponds to 
the diagonal elements
$a_{k,k}$ and $\ell=1$ corresponds to the elements $a_{k+1,k}$ of the 
first subdiagonal, etc.

The f-KT hierarchy consists of the symmetries of the f-KT lattice generated by the Chevalley invariants $I_k=\frac{1}{k+1}\text{tr}(L^{k+1})$ for $k=1,\ldots,M-1$ (c.f. \cite{KS08}).
Each symmetry is given by the equation
\begin{equation}\label{KT-hierarchy}
\frac{\partial L}{\partial t_k}=[(L^k)_{\ge0}, L],\qquad\text{for}\quad k=1,2,\ldots, M-1.
\end{equation}
In this paper, we formally use a set of infinitely many symmetries, whose flow parameter is denoted by
$\mathbf{t}:=(t_1,t_2,\ldots)$. We call the set of all symmetries in the form \eqref{KT-hierarchy}
the f-KT lattice hierarchy, or simply f-KT hierarchy.

\subsection{Matrix factorization and the $\tau$-functions of the f-KT lattice}
To find the solution $L(t)$ of the f-KT lattice in terms of the initial matrix $L(0)$,
it is standard to consider the \emph{LU-factorization} (Gauss decomposition) of the matrix $\exp(tL(0))$:
\begin{equation}\label{exp}
\exp(tL(0))=u(t)b(t)\qquad \text{with}\quad u(t)\in \mathcal{N}^-,~ b(t)\in \mathcal{B}^+. 
\end{equation}
where $\mathcal{N}^-$ is the set of lower triangular matrices with $1$'s on the diagonal, and
$\mathcal{B}^+$ is the set of invertible upper triangular matrices.
Note here that $u(0)=b(0)=I$, the identity matrix.
It is known and easy to show that this factorization exists if and only if the principal leading minors are all nonzero.  
Under this assumption, we let $\Delta_N(t)$ denote the $N$-th leading principal minor, the determinant of the $N\times N$ upper left corner of the matrix $\exp(tL(0))$,
\begin{equation}\label{eq:tau1}
\Delta_N(t):=[\exp(tL(0))]_N\qquad\text{for}\quad N=1,2,\ldots,M.
\end{equation}
Let $b(t)=h(t)n(t)$ be the decomposition with $h(t)$ a diagonal matrix and $n(t)\in \mathcal{N}^+$, a unipotent upper triangular matrix. Then writing $h(t)=\exp(\sum_{i=1}^Mc_i(t)E_{i,i})$, we have 
\begin{equation}\label{eq:ci}
\Delta_N(t)=\exp\left(\sum_{i=1}^Nc_i(t)\right),\quad\text{which gives}\quad c_N(t)=\ln\frac{\Delta_N(t)}{\Delta_{N-1}(t)}.
\end{equation}
This formula gives a key to express the solution of the f-KT lattice as explained below.

With the factorization \eqref{exp}, we first have the following result (c.f. \cite{Symes} which deals with the
original Toda lattice).
\begin{proposition}\label{prop:LUsolution}
The solution $L(t)$ of \eqref{Lax} is given by
\begin{equation}\label{eq:ub}
L(t)=u^{-1}(t)L(0)u(t)=b(t)L(0)b(t)^{-1}.
\end{equation}
\end{proposition}
\begin{proof}
Taking the derivative of \eqref{exp}, we have
\[
\frac{d}{dt}\exp(tL(0))=L(0)ub=ubL(0)=\dot{u}b+u\dot{b},
\]
where $\dot{x}$ means the derivative of $x(t)$.  This 
equation can  also be written as
\[
u^{-1}L(0)u=bL(0)b^{-1}=u^{-1}\dot{u}+\dot{b}b^{-1}.
\]
We denote this as $\tilde{L}$, and show $\tilde{L}=L$.  
Using $\mathfrak{g}=\mathfrak{n}_-\oplus\mathfrak{b}_+$, we decompose $\tilde{L}$ 
as
\begin{equation}\label{eq:b-eq}
u^{-1}\dot{u}=(\tilde{L})_{<0}
\qquad\text{and}\qquad
\dot{b}b^{-1}=(\tilde{L})_{\ge 0}.
\end{equation}
To show $\tilde L=L$, we first show that $\tilde{L}$ also satisfies the f-KT lattice.  Differentiating $\tilde{L}=bL(0)b^{-1}$, we have
\[
\frac{d\tilde{L}}{dt}=\dot{b}b^{-1}\tilde{L}-\tilde{L}\dot{b}b^{-1}=[\dot{b}b^{-1},\tilde{L}] = [(\tilde L)_{\ge0},\tilde L].
\]
Here we have used $\frac{d}{dt}b^{-1}=-b^{-1}\dot{b}b^{-1}$.
Since $\tilde L(0)=L(0)$, i.e. the initial data are the same,  the uniqueness theorem of the differential equation implies that $\tilde L(t)=L(t)$.  This completes the proof.
\end{proof}

We also have an explicit formula for the diagonal elements of $L(t)$.
\begin{proposition}\label{prop:diagonal}
The diagonal elements of the matrix $L=L(t)$ can be expressed by
\begin{equation}\label{aii}
a_{N,N}(t)=\frac{d}{dt}\ln\frac{\Delta_N(t)}{\Delta_{N-1}(t)}, \qquad 1 \le N \le M,
\end{equation}
where $\Delta_N(t)=[\exp(tL(0))]_N$ is defined in \eqref{eq:tau1}.
\end{proposition} 
\begin{proof}
First recall from the proof of Proposition \ref{prop:LUsolution}
that $\dot{b}b^{-1}=(L)_{\ge 0}$, i.e.
\[
(L)_{\ge 0}=\sum_{i=1}^Ma_{i,i}E_{i,i}+\sum_{i=1}^{M-1}E_{i,i+1}.
\]
The decomposition $b(t)=h(t)n(t)$ and $h(t)=\exp(\sum_{i=1}c_i(t)E_{i,i})$  gives
\[
\dot{b}b^{-1}=\dot{h}h^{-1}+h\dot{n}n^{-1}h^{-1}, \quad\text{where}\quad \dot{h}h^{-1}=\sum_{i=1}^M\dot{c}_iE_{i,i}.
\]
This completes the proof.
\end{proof}

For the hierarchy, we consider the notation with multi-time variables $\mathbf{t}=(t_1,t_2,\ldots)$
\begin{equation} \label{eq:Theta}
\Theta_{A}(\mathbf{t}):=\sum_{n=1}^{\infty} A^n t_n\qquad\text{for a matrix}~ A.
\end{equation}

Then Proposition \ref{prop:LUsolution} 
can be easily extended  to the f-KT hierarchy.
\begin{proposition}\label{prop:LUsolution2}
Consider the LU-factorization
\begin{equation}\label{eq:ubhky}
\exp(\Theta_{L^0}(\mathbf{t}))=
u(\mathbf{t})b(\mathbf{t})\qquad \text{with}
\quad u(\mathbf{t})\in \mathcal{N}^-,~ b(\mathbf{t})\in \mathcal{B}^+. 
\end{equation}
The solution $L(\mathbf{t})$ of the f-KT hierarchy is then given by
\[
L(\mathbf{t})=u(\mathbf{t})^{-1} L^0u(\mathbf{t})=
b(\mathbf{t})L^0b(\mathbf{t})^{-1},
\]
where $L^0=L(\mathbf{0})$, the initial matrix of $L(\mathbf{t})$.
\end{proposition}

This leads to the following definition of the minors $\Delta_N, 1 \le N \le M$ for the f-KT hierarchy,
\begin{equation}\label{eq:tau2}
\Delta_N(\mathbf{t})=[\exp(\Theta_{L^\0}(\mathbf{t}))]_N\qquad\text{with}\quad
\mathbf{t}=(t_1,t_2,\ldots).
\end{equation}

As before, the LU-factorization
\eqref{eq:ubhky}
exists if and only if each $\Delta_k(\mathbf{t})$ is nonzero.

In this paper, we study the cases where some of the minors $\Delta_N(\t)$ vanishes, and the set of zeros $\t_*$ with $\Delta_N(\t_*)=0$ is referred to as the Painlev\'e divisor in \cite{FH}.

\begin{remark}
If one identifies $t_1=x,~t_2=y$ and $t_3=t$,  the determinant $\Delta_N(\t)$ in
\eqref{eq:tau2} gives so-called $\tau$-function for the KP equation, which plays a key role in classifying solutions
(c.f. \cite{K:17} and below). There the $\Delta_N$ is 
associated with a point of the Grassmannian $\text{Gr}(N,M)$, 
and the set $(\Delta_1,\ldots,\Delta_{M-1})$ is associated with a point of the flag variety of $\text{GL}_M$.
The solution space of the f-KT hierarchy is naturally given by the complete flag variety of type-$A$ in this case (c.f. \cite{Xie:22b} and below).
\end{remark}


\subsection{The Schur functions and the $\tau$-functions}
One of the main objectives in this paper is to discuss the Schur function expansion of the minors $\Delta_N(\t)$.
The Schur functions are generated by the functions $p_n(\t)$, sometimes called the elementary
Schur polynomials. These $p_n(\t)$ are defined by
\[
\exp\left(\sum_{i=1}^\infty k^it_i\right)=\sum_{n=0}^\infty p_n(\t)k^n,
\]
which gives
\[
p_n(\t)=\sum_{i_1+2i_2+\cdots+ni_n=n}\frac{t_1^{i_1}t_2^{i_2}\cdots t_n^{i_n}}{i_1!\,i_2!\,\cdots\,i_n!}.
\]
The functions $p_n(\t)$ satisfy
\[
{\partial_k}p_n(\t)={\partial_1^k}p_n(\t)=p_{n-k}(\t),
\]
where $\partial_i^j=\frac{\partial^j}{\partial t_i^j}$, and we also have $p_i(\t)=0$ if $i<0$.

Let $\lambda$ be a Young diagram with $\lambda=(\lambda_1,\lambda_2,\ldots,\lambda_N)$
where $\lambda_1\ge\lambda_2\ge\cdots\ge\lambda_N\ge 0$.
The Schur function associated with the Young diagram $\lambda=(\lambda_1,\lambda_2,\ldots,\lambda_N)$ is define by
\[
S_\lambda(\t)=\text{det}\left(p_{\lambda_i-i+j}(\t)\right)_{1\le i,j\le N}.
\]
Note that the $S_\lambda(\t)$ can be written in the Wronskian form with respect to $t_1$,
\begin{equation}\label{eq:WrSchur}
S_\lambda(\t)=\text{Wr}\left(p_{\lambda_N},p_{\lambda_{N-1}+1},\ldots,p_{\lambda_1+N-1}\right)=\text{det}\left(p_{\lambda_{N-j+1}-i+j}(\t)\right)_{1\le i,j\le N}.
\end{equation}

The minors $\Delta_N(\t)$ in \eqref{eq:tau2} can also be expressed in the Wronskian form. 
First recall that there exists a unique $n_0\in \mathcal{N}^-$ such that
\begin{equation}\label{eq:VK}
L^0=n_0^{-1}C_{K} n_0,
\end{equation}
where $C_{K}$ is the companion matrix associated with $L^0$, i.e.
\[
C_{K}=\begin{pmatrix}
0 & 1 & 0 &\cdots &0\\
0& 0& 1 &\cdots &0\\
\vdots &\vdots &\ddots &\ddots &\vdots\\
0& 0 &\cdots & 0 & 1\\
(-1)^M \sigma_M&(-1)^{M-1}\sigma_{M-1}&\cdots &\sigma_2& -\sigma_1
\end{pmatrix},
\]
where the $\sigma_i$'s are defined by det$(\kappa I-L^0)=\sum_{i=0}^M(-1)^i\sigma_i\kappa^{n-i}=\prod_{i=1}^M(\kappa-\kappa_i)$, that is, $\sigma_i$ is the elementary symmetric polynomial of degree $i$ of
$\{\kappa_1,\ldots,\kappa_M\}$.  
We also have $C_{K}=VK V^{-1}$ when $K=\text{diag}(\kappa_1,\ldots,\kappa_M)$ is invertible, where $V$ is the Vandermonde matrix $V=(\kappa_j^{i-1})_{1\le i,j\le M}$.
Then we write $\Delta_N(\t)$ in the determinant form,
\begin{align*}
\Delta_N(\t)&=\langle e_1\wedge\cdots\wedge e_N,\,n_0^{-1}e^{\Theta_{C_K}(\t)}n_0\cdot 
e_1\wedge\cdots\wedge e_N\rangle=\text{det}\left(\langle e_i,\, e^{\Theta_{C_K}(\t)} n_0\cdot  e_j\rangle\right)_{1\le i,j\le M}.
\end{align*}
Then setting,
\[
f_j(\t):=\langle e_1,\,e^{\Theta_{C_K}(\t)}n_0\cdot e_j\rangle =\langle e_1,\,Ve^{\Theta_K(\t)}V^{-1}n_0\cdot e_j\rangle \qquad \text{for}\quad j=1,\ldots, N,
\]
we have
\[
\langle e_i,\,Ve^{\Theta_K(\t)}V^{-1}n_0\cdot e_j\rangle=\partial_1^{i-1}f_j(\t)\qquad\text{for}\quad i=1,\ldots,N.
\]
Thus, each $\Delta_N(\t)$ can be expressed in the Wronskian form,
\begin{equation}\label{eq:Wrtau}
\Delta_N(\t)=\text{Wr}(f_1,f_2,\ldots,f_N).
\end{equation}
One should also note that each $f_j(\t)$ satisfies the heat hierarchy,
\[
\partial_n f_j(\t)=\partial_1^n f_j(\t)\qquad n=1,2,\ldots.
\]

\begin{remark}
The formula \eqref{eq:Wrtau} is the $\tau$-function of the KP hierarchy (c.f. \cite{K:17}).
Also as shown below, the minors $\Delta_N(\t)$ for $1\le N\le M-1$ give the $\tau$-functions of the f-KT hierarchy of
type $A$.
\end{remark}

Let $\{l_0,l_1,\ldots, l_{N-1}\}$ be a set of integers with $0\le l_0<l_1<\cdots<l_{N-1}$. Then we define
\begin{equation}\label{e:f-Determinant}
[l_0,l_1,\ldots,l_{N-1}]:=\left|\begin{array}{cccc}
f_1^{(l_0)} & f_1^{(l_1)} & \cdots & f_1^{(l_{N-1})} \\
f_2^{(l_0)} & f_2^{(l_1)} &\cdots & f_2^{(l_{N-1})} \\
\vdots & \vdots& \ddots & \vdots \\
f_N^{(l_0)}& f_N^{(l_1)} &\cdots & f_N^{(l_{N-1})}
\end{array}\right|,
\end{equation}
where $f_j^{(l)}=\partial_1^lf_j$.
Note that the determinant $\Delta_N(\t)$ in \eqref{eq:Wrtau} is then given by
\[
\Delta_N(\t)=[0,1,2,\ldots,N-1].
\]
We define the Young diagram $\lambda$  associated with the set $\{l_0,l_1,\ldots,l_{N-1}\}$ as follows:
The southeast boundary of the Young diagram gives a path labeled by consecutively increasing integers from $0$, so that the indices $\{l_0,l_1,\ldots,l_{N-1}\}$ appear on the vertical edges in the boundary path as shown below.
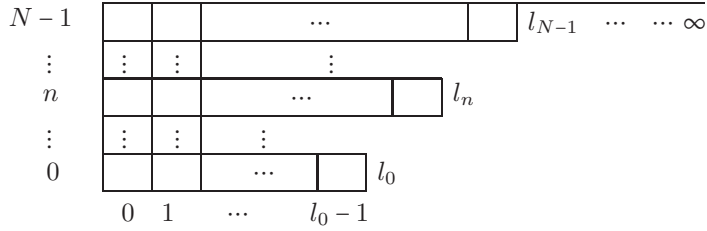
\setlength{\unitlength}{0.5mm}
\begin{center}
  \begin{picture}(160,60)
\put(5,55){\line(1,0){160}}
\put(5,45){\line(1,0){110}}
  \put(5,35){\line(1,0){90}}
  \put(5,25){\line(1,0){90}}
  \put(5,15){\line(1,0){70}}
  \put(5,5){\line(1,0){70}}
  \put(5,5){\line(0,1){50}}
 \put(31,5){\line(0,1){50}}
  \put(18,5){\line(0,1){50}}
  \put(75,5){\line(0,1){10}}
  \put(62,5){\line(0,1){10}}
  \put(95,25){\line(0,1){10}}
  \put(82,25){\line(0,1){10}}
\put(115,45){\line(0,1){10}}
\put(102,45){\line(0,1){10}}
  \put(-11,29){$n$}
  \put(98,29){${l_{n}}$}
  \put(-20,49){$N-1$}
  \put(-10,37){$\vdots$}
  \put(-10,17){$\vdots$}
  \put(65,37){$\vdots$}
  \put(10,37){$\vdots$}
  \put(24,37){$\vdots$}
  \put(60,48){$\cdots$}
  \put(55,28){$\cdots$}
  \put(47,17){$\vdots$}
  \put(10,17){$\vdots$}
  \put(24,17){$\vdots$}
  \put(-10,8){$0$}
  \put(45,8){$\cdots$}
  \put(78,8){${l_0}$}

  \put(118,48){${l_{N-1}}\quad{\cdots}\quad {\cdots~\infty}$}
 
    \put(10,-3){${0\quad1}\qquad {\cdots}$}
    \put(60,-3){${l_0-1}$}
 \end{picture} 
\end{center}

\medskip
\noindent
The Young diagram $\lambda=(\lambda_1,\lambda_2,\ldots,\lambda_N)$ is expressed by
\[
\lambda_k=l_{N-k}-(N-k)\qquad\text{for}\quad  1\le k\le N.
\]
Let $\Delta_\lambda(\t)$ denote the determinant \eqref{e:f-Determinant}.
Then the $\Delta_N(\t)$ in \eqref{eq:Wrtau} is denoted by
$\Delta_N(\t)=\Delta_{\emptyset}(\t)$, since $(l_0,l_1,\ldots,l_{N-1})=(1,2,\ldots,N)$.
We then have the following useful formula (c.f. \cite{K:17}),
\begin{equation}\label{eq:tau-Schur}
S_\lambda(\tilde \partial)\Delta_\emptyset(\t)=\Delta_{\lambda}(\t),
\end{equation}
where $\tilde \partial$ is defined by
\[
\tilde\partial=\left(\partial_1,\,\frac{1}{2}\partial_2,\,\frac{1}{3}\partial_3,\,\ldots\right).
\]
In particular, if we choose $\Delta_\emptyset(\t)=S_\lambda(\t)$ as in \eqref{eq:WrSchur},
we have
\begin{align}\label{eq:skewS}
S_\mu(\tilde\partial)\,S_\lambda(\t)&=\text{det}\left(p_{\lambda_{N-i+1}-\mu_{N-j+1}+i-j}\right)_{1\le i,j\le N}\\ \nonumber
&=\text{det}\left(p_{\lambda_i-\mu_j-i+j}\right)_{1\le i,j\le N}=: S_{\lambda/\mu}(\t),
\end{align}
which is the {\it skew} Schur polynomial. Note here that $S_{\lambda/\mu}\ne0$ iff $\mu\subseteq\lambda$,
and in particular, $S_{\lambda/\mu}=1$ when $\lambda=\mu$. This then defines the orthogonality,
\[
S_\mu(\tilde\partial)S_\lambda(\t)\Big|_{\t=\0}=\delta_{\lambda,\mu}.
\]

We then have the Schur function expansion of $\Delta_N(\t)$ in the form (c.f. \cite{K:17}),
\begin{equation}\label{eq:Sexpansion}
\Delta_\emptyset (\t+\mathbf{s})=\sum_{\lambda}\Delta_{\lambda}(\mathbf{s})S_\lambda(\t),
\end{equation}
where the sum is over all possible Young diagrams. Note here that the coefficients $\Delta_\lambda(\mathbf{s})$ 
are all determinants and satisfy all the Pl\"ucker relations (this is an important result of the Sato theory of the KP hierarchy, c.f. \cite{K:17}).

In this paper, we are interested in the points $\t_*$ so that for some $\lambda_0$,
\[
\Delta_\mu(\t_*)=0\qquad\text{for}\qquad \mu<\lambda_0.
\]
hence, the $\Delta_\emptyset(\t)$ has the expansion,
\begin{equation}\label{eq:expansion}
\Delta_\emptyset(\t+\t_*)=S_{\lambda_0}(\t)+\sum_{\lambda_0<\lambda}\Delta_\lambda(\t_*)S_\lambda(\t).
\end{equation}
Here we have normalized $\Delta_{\lambda_0}(\t_*)=1$.

\section{The Schur expansion of $\Delta_N(\t)$}\label{sec:Schurexpand}
If a leading principal minor of the matrix $\exp(\Theta_{L^0}(\t))$ vanishes at a fixed multi-time $\t=\t_*$, then 
the LU-factorization of $\exp(\Theta_{L^0}(\mathbf{t}_*))$ fails, i.e.
$\exp(\Theta_{L^0}(\mathbf{t}_*))\in \mathcal{N}^-\dot{w}\mathcal{B}^+$ for some $w \in \Sym_M$, the symmetric group of order $M$,
such that $w \neq e$ and $\dot w$ denotes a representative of $w$ in $SL_M$.
This means that the f-KT flow becomes singular 
and hits the  boundary of a Schubert cell marked by $w\in\Sym_M$
(c.f. \cite{CK, FH}). The set of times $\t_*$  is called the \emph{Painlev\'e divisor}.  

Now we express the $\tau$-functions near the Painlev\'e divisors. Set $\t\to \t+\t_*$, where
$\Delta_N(\t_*)=0$ for some $1\le N\le M-1$. Then we have
\begin{equation}\label{eq:nwb}
\exp(\Theta_{L^0}(\t_*))=n_*\dot w_*b_*\qquad\text{for some}\quad w_*\in \frak{S}_M,
\end{equation}
where $n_*\in \mathcal{N}^-$ and $b_*\in \mathcal{B}^+$. In particular, the unipotent matrix $n_*$ for each $w_*$ will be constructed explicitly below.

 Then, using $L^0=n_0^{-1}VKV^{-1}n_0$ (see \eqref{eq:VK}), we have
\begin{equation}\label{eq:DN}
\Delta_N(\t)=d_N\langle e_1\wedge\cdots\wedge e_N, Ve^{\Theta_{K}(\t)}V^{-1}{n_0n_*}\dot w_* \cdot e_1\wedge\cdots \wedge e_N\rangle,
\end{equation}
where the constant $d_N$ is given by $b_* \cdot e_1\wedge\cdots\wedge e_N= d_N e_1\wedge\cdots\wedge e_N$, but
note that the solution of the hierarchy does not depend on $d_N$.

Now expanding $Ve^{\Theta_{K}(\t)}$, we have
\begin{equation}\label{eq:P}
Ve^{\Theta_{K}(\t)}=
\begin{pmatrix}
1 & p_1 & p_2 & \cdots &\cdots &\cdots &\cdots\\
0&1&p_1& p_2&\cdots &\cdots &\cdots\\
\vdots &\ddots &\ddots &\ddots &\ddots &\ddots &\vdots\\
0&\cdots &0&1 & p_1 &p_2 &\cdots
\end{pmatrix}\, V_{\infty}=: \mathcal{P}(\t)V_\infty,
\end{equation}
where $V_\infty$ is the $\infty\times M$ Vandelmonde matrix $V_\infty=(\lambda_j^{i-1})$ for
$1\le i$ and $1\le j\le M$.  We also note that
\[
(V_\infty V^{-1})^T=\begin{pmatrix}
1&0& \ldots &0&(-1)^{M-1} c_{(0|M-1)} & (-1)^{M-1}c_{(1|M-1)}& (-1)^{M-1}c_{(2|M-1)}&\cdots\\
0&1&\cdots& 0&(-1)^{M-2}c_{(0|M-2)}&      (-1)^{M-2}c_{(1|M-2)}& (-1)^{M-2}c_{(2|M-2)}&\cdots\\
\vdots&\vdots&\ddots&\vdots &\vdots &\vdots &\vdots&\vdots \\
0& 0& \cdots & 1 & c_{(0|0)} & c_{(1|0)} &c_{(2|0)} & \cdots
\end{pmatrix},
\]
where $c_{(a|b)}=c_{(a|b)}(\lambda)$ is the Schur polynomial of $\lambda=(\lambda_1,\ldots,\lambda_n)$ associated with the hook Young diagram $\lambda$ of $(a|b)$ type, i.e. $\lambda=(1+a,\underbrace{1,\ldots,1}_{b})$,  e.g. $(3|2)$ is $\young[4,1,1][5]$, $\lambda=(4,1,1)$. This is nothing but the Giambelli formula.  

For each $w_* \in \frak{S}_M$, we note
\[
\dot w_*\cdot e_1\wedge e_2\wedge \cdots\wedge e_N= e_{w_*(1)}\wedge e_{w_*(2)}\wedge
\cdots\wedge e_{w_*(N)}.
\]
Then each $\infty$-vector $v_{w_*(j)}:=V_\infty V^{-1}n_0n_*\dot w_* e_j$ has the following structure
\begin{equation}\label{eq:v}
v_{w_*(j)}=(0,\ldots,0,1,*,*,\cdots)^T,
\end{equation}
where the first nonzero term with 1 is located at $w_*(j)$.  We then have the following formula 
of the minor $\Delta_N(\t)$,
\begin{equation}\label{tauW}
\Delta_N(\t; w_*)= \langle e_1\wedge \cdots\wedge e_N, \mathcal{P}(\t)\cdot v_{w_*(1)}\wedge\cdots\wedge
v_{w_*(N)}\rangle.
\end{equation}

Now we have the following theorem.
\begin{theorem}\label{thm:Schurexpand}
The $\tau$-function in \eqref{tauW} has the  Schur expansion,
\[
\Delta_N(\t; w_*)=S_{\lambda}(\t)+\sum_{\mu\supset\lambda}c_{\mu}S_\mu(\t),
\]
where $S_\lambda(t)$ is the Schur polynomial associated with the Young diagram $\lambda$, which
is determined by $w_*$ as follows: Let $\{i_1,i_2,\ldots,i_N\}$ be defined by
\[
\{i_1,i_2,\ldots,i_N\}=\{w_*(1),w_*(2),\ldots, w_*(N)\},
\]
with $i_1<i_2<\cdots<i_N$. Then the Young diagram $\lambda=(\lambda_1,\lambda_2,\ldots,\lambda_N)$ is given by
\[
\lambda_k=i_k-k\qquad\text{for}\qquad 1\le k\le N.
\]
\end{theorem}

\begin{proof}
Apply the Cauchy-Binet lemma to expand $\tau_N(\t)$.
\end{proof}

The number of free parameters in the $\tau$-function is given by the dimension of the corresponding Schubert cell, which is defined by
\[
\text{dim}(\mathcal{N}^-\dot w_*\mathcal{B}^+/\mathcal{B}^+)=\ell( w_0)-\ell(w_*),
\]
where $ w_0$ is the longest element of $\frak{S}_M$.
Here the spectrum (eigenvalues) are fixed.
 In particular, if we fix all the eigenvalues to be zero, i.e. in the nilpotent limit, we have the polynomial solutions,
\begin{equation}\label{eq:poly}
\Delta_N(\t; w_*)=\langle e_1,\wedge\cdots\wedge e_N, \mathcal{P}_M(\t) n_*\dot w_* \cdot e_1\wedge\cdots\wedge e_{N}\rangle,
\end{equation}
where $\mathcal{P}_M(\t)$ is the $M\times M$ matrix associated with a regular nilpotent matrix $L^0=\mathsf{N}$ in \eqref{eq:pnil}, 
\begin{equation}\label{eq:polyM}
\mathcal{P}_M(\t)=e^{\Theta_{\mathsf{N}}(\t)}=\sum_{k=0}^{M-1}p_k(\t)\mathsf{N}^k\qquad\text{with}\qquad \mathsf{N}=\sum_{i=1}^{M-1}E_{i,i+1}.
\end{equation}

We note the following property of the Schur expansion of $\Delta_N(\t;w_*)$.
Let $[N]$ be the set of numbers given by $\{1,2,\ldots, N\}$. Then we have the following lemma.
\begin{lemma}
If $w_*[N]=[N]$, then $\Delta_N(\mathbf{0}; w_*)=\pm1$.
\end{lemma}
\begin{proof}
The minor $\Delta_N(\t)$ is given by
\begin{align*}
\Delta_N(\t; w_*)&=\langle e_1\wedge\cdots\wedge e_N,  \mathcal{P}(\t)\cdot v_{w_*(1)}\wedge\cdots\wedge
v_{w_*(N)}\rangle\\
&=\pm \langle e_1\wedge\cdots\wedge e_N,  \mathcal{P}(\t)\cdot v_1\wedge\cdots\wedge
v_N\rangle=\pm 1+\text{h.o.t.},
\end{align*}
where we recall $v_j=(0,\cdots,0,\overset{j}{1},*\cdots)^T$.
\end{proof}
Notice that this lemma can be restated as $\Delta_N(\mathbf{0}; w_*)\ne 0$ if and only if $w_*[N]\ne [N]$.
 From this lemma, the following is obvious.
\begin{corollary}
If $w_0$ is the longest element in $\frak{S}_M$, then $\Delta_N({\bf 0}; w_0)=0$ for
all $N=1,\ldots, M-1$.
\end{corollary}


\section{The $\ell$-KT hierarchy}\label{sec:lbanded}

For an explicit expression of the elements $a_{i,j}$,
we have the following proposition, which can be found in \cite{AvM, KY} (see also Appendix A).
\begin{proposition}\label{prop:aij}
The elements $a_{i,j}(\t)$ for the f-KT hierarchy can be expressed as
\begin{equation}\label{aij}
a_{i,j}(\t)=\frac{p_{i-j+1}(\tilde D)\Delta_j(\t)\circ \Delta_{i-1}(\t)}{\Delta_j(\t)\Delta_{i-1}(\t)} \qquad\text{for}\quad 1\le j\le i\le M,
\end{equation}
where $\tilde D$ is defined by
\[
\tilde D:=\left(D_1,\,\frac{1}{2}D_2,\,\frac{1}{3}D_3,\ldots\,\right)\qquad\text{with}\qquad 
D_nf\circ g=\left(\frac{\partial}{\partial t_n}-\frac{\partial}{\partial t'_n}\right)f(\t)g(\t')\Big|_{\t=\t'}.
\]
\end{proposition}

We consider the special case of $L$, which has the following banded structure: For each $\ell+1\le m\le M-1$, the elements in the lower triangular part satisfy $a_{k+m,k}(\t)=0$ for $1\le k\le  M-m$.
That is, the subdiagonals lower than the $\ell$-th diagonal are all zero. We call the f-KT hierarchy
with this banded Lax matrix the $\ell$-KT hierarchy (see below for the details). Note that the 1-banded KT hierarchy is the
original tridiagonal Kostant-Toda (KT) lattice hierarchy, and the $(M-1)$-banded KT hierarchy is the f-KT
hierarchy.

Then we have the following proposition:
\begin{proposition}\label{prop:l-band}
The elements $a_{i,j}$ in the $\ell$-th diagonal of the $\ell$-KT hierarchy are given by
\[
a_{k+\ell, k}(\t)=a_{k+\ell,k}^0\frac{\Delta_{k+\ell}(\t)\Delta_{k-1}(\t)}{\Delta_{k+\ell-1}(\t)\Delta_k(\t)}\qquad\text{for}\quad 1\le k\le M-\ell,
\]
where $a_{k+\ell,k}^0$ is a constant.
\end{proposition}
\begin{proof}
Since $a_{\ell+k+1,k}=a_{\ell+k,k-1}=0$ for all $1\le k\le M-\ell$, the equations
 \eqref{eq:f-KTcomp} become
 \[
 \frac{da_{k+\ell,k}}{dt}=(a_{k+\ell,k+\ell}-a_{k,k})a_{k+\ell,k}.
 \]
Using the expression \eqref{aii} for $a_{i,i}$, we can integrate the equations and obtain the formula
in the proposition.
\end{proof}

Note here that the constants $a_{k+\ell,k}^0$ can be normalized as 1, if it is not zero. 
Using Proposition \ref{prop:aij}, we have the bilinear equation for the $\ell$-KT hierarchy,
\begin{equation}\label{l-band}
p_{\ell+1}(\tilde D)\Delta_k\circ \Delta_{k+\ell-1}(\t)=a_{k+\ell,k}^0\Delta_{k+\ell}(\t)\Delta_{k-1}(\t)\qquad\text{for}\quad 1\le k\le M-\ell.
\end{equation}

When $\ell=1$, we see the bilinear equation of the original tridiagonal Toda lattice,
\[
p_2(\tilde D)\Delta_k\circ\Delta_k=\frac{1}{2}D_1^2\Delta_k\circ\Delta_k=\Delta_k\partial_1^2\Delta_k-(\partial_1\Delta_k)^2=a_{k+1,k}^0\Delta_{k+1}\Delta_{k-1}.
\]


\subsection{The $\tau$-functions of the $\ell$-KT hierarchy}
We start with the following proposition on some properties of the $\ell$-th diagonal elements in the Lax matrix.
\begin{proposition}\label{prop:lbandKT}
Let $w_*$ be an element in $\frak{S}_M$ satisfying the following conditions,
\begin{itemize}
\item[(a)]
$w_*(i)\ne w_*(M-j+1)+1\quad\text{for any}\quad 2\le i+j\le M-\ell,$ and
\item[(b)] for some $1\le i\le M-\ell$, there is a relation,
\[
w_*(i)=w_*(i+\ell)+1.
\]
\end{itemize}
Then the corresponding $\tau$-functions $\Delta_N(\t; w_*), 1 \le N \le M-1$ give a solution to the $\ell$-KT hierarchy, that is,
for each $m$ with $\ell+1\le m\le M-1$, the elements $a_{i,j}$ in the Lax matrix $L$ satisfy
\[
a_{m+k,k}(\t)=0\qquad\text{for}\qquad 1\le k\le M-m,
\]
and $a_{\ell+k,k}(\t)\ne 0$ for  some $1\le k\le M-\ell$.
\end{proposition}

\begin{proof}
We first note that $a_{M,1}(\t)$ satisfies 
\[
\frac{\partial a_{M,1}}{\partial t_1}=(a_{M,M}-a_{1,1})\,a_{M,1},
\]
which gives
\[
a_{M,1}(\t)=a_{M,1}^0\frac{\Delta_M\Delta_0}{\Delta_{M-1}(\t)\Delta_1(\t)},
\]
where $\Delta_0=1$ and $\Delta_M$ is a constant.  From \eqref{l-band}, we have
\begin{equation}\label{eq:pn}
p_{M}(\tilde D)\Delta_{1}\circ\Delta_{M-1}=a_{M,1}^0\Delta_{M}\Delta_0=\text{constant}.
\end{equation}
The $\Delta_1$ and $\Delta_{M-1}$ are given by
\begin{align*}
\Delta_1(\t)&=\langle e_1,  \mathcal{P}(\t)\cdot  v_{w_*(1)}\rangle=S_{\mu_1}(\t)+\text{h.o.t},\\[1.0ex]
\Delta_{M-1}(\t)&=\langle e_1\wedge\cdots\wedge e_{M-1},  \mathcal{P}(\t)\cdot v_{ w*(1)}\wedge\cdots
\wedge v_{w_*(M-1)}\rangle=S_{\mu_{M-1}}(\t)+\text{h.o.t}.,
\end{align*}
where the Young diagram $\mu_1$ and $\mu_{M-1}$ are given by
\begin{align*}
S_{\mu_1}(\t)&=p_{|\mu_1|}(\t)\qquad\text{with}\quad |\mu_1|=w_*(1)-1,\\
S_{\mu_{M-1}}(\t)&=(-1)^{|\mu_{M-1}|}p_{|\mu_{M-1}|}(-\t)\qquad \text{with}\quad |\mu_{M-1}|=M-w_*(M),
\end{align*}
that is, $\mu_1$ consists of $|\mu_1|$ horizontal boxes, and $\mu_{M-1}$ consists of $|\mu_{M-1}|$
vertical boxes.

The degree of the left hand side of \eqref{eq:pn} is given by
\[
\text{deg}(\Delta_1)+\text{deg}(\Delta_{M-1})-M=(w_*(1)-1)+(M-w_*(M))-M=w_*(1)-w_*(M)-1.
\]
Since the right hand side of \eqref{eq:pn} has degree zero, if $w_*(1)\ne w_*(M)+1$, then
$a_{M,1}^0$ must be zero. This confirms the case (a) with $i=j=1$ and $\ell=M-2$.

Suppose that $a_{M-a+k,k}=0$ for all $1\le a\le m-1$ and $1\le k\le a$.  Then we have
\begin{equation}\label{eq:pn1}
p_{M-m+1}(\tilde D)\Delta_k\circ\Delta_{M-(m-k)-1}=a_{M-m+k,k}^0\Delta_{k-1}\Delta_{M-(m-k)}.
\end{equation}
The order of the left hand side of \eqref{eq:pn1} is then given by
\[
\left(k_1+\cdots+k_k-\frac{k(k+1)}{2}\right)+\left(M(m+1-k)-(l_1+\cdots+l_{m+1-k})-\frac{(m-k)(m-k+1)}{2}\right)-(M-m+1),
\]
where we denote $w_*(i)=k_i$ and $w_*(M-j+1)=l_j$. The order of the right hand side of \eqref{eq:pn1} is
\[
\left(k_1+\ldots+k_{k-1}-\frac{k(k-1)}{2}\right)+\left(M(m-k)-(l_1+\cdots+l_{m-k})-\frac{(m-k)(m-k-1)}{2}\right).
\]
The condition that these orders are different is given by
\[
k_k\ne l_{m+1-k}+1,\qquad\text{which is}\qquad w_*(k)\ne w_*(M-m+k)+1.
\]
So if this is true, then we have $a_{M-m+k,k}^0=0$. That is, $a_{M-m+k,k}(\t)=0$.
Continuing this argument up to $m=M-\ell-1$ gives the proof.
\end{proof}

\begin{remark}
The two conditions on $w_*$ in Proposition \ref{prop:lbandKT} have the following more symmetric form
\begin{itemize}
\item[(a)]
$w_*(i)\ne w_*(i+k)+1\quad\text{for any } k > \ell,$ and
\item[(b)] for some $1\le i\le M-\ell$, there is a relation,
\[
w_*(i)=w_*(i+\ell)+1.
\]
\end{itemize}
Condition (a) means that no elements on the $k$-th off-diagonal is allowed to be nontrivial for $k > \ell$ and condition (b) means some elements on the $\ell$-th off-diagonal are non-trivial. These two conditions can be used independently. For example if the definition of $\ell$-KT hierarchy is inclusive, that is $\ell - 1$-banded KT hierarchy is thought to be a sub of the $\ell$-KT hierarchy, then only condition (a) is needed to be satisfied for all such $w_*$.
\end{remark}

\begin{corollary}
Let $w_*=w_0$ be the longest element in $\frak{S}_M$. Then the corresponding $\tau$-functions give a solution to the classical tridiagonal KT hierarchy.
\end{corollary}
\begin{proof}
It is obvious that $w_0$ satisfies the conditions in Proposition \ref{prop:lbandKT} for $\ell = 1$. 
\end{proof}


\section{The f-KT hierarchy on simple Lie algebras}\label{sec:Lie}
The definition of f-KT lattice (hierarchy) in Section \ref{sec:fKT} in the Hessenberg form can be extended to a system defined on some simple Lie algebras.  
Let $\G$ be a connected simply connected complex semisimple Lie group with Lie algebra $\text{Lie}(\G) = \mathfrak{g}$.
Let $\{H_i, X_i, Y_i\}$ be the Chevalley basis of the Lie algebra $\frak{g}$ of rank $n$, which satisfy
\[
[H_i, H_j]=0,\qquad [H_i, X_j]=C_{j,i}X_j,\qquad [H_i, Y_j]=-C_{j,i}Y_j,\qquad [X_i,Y_j]=\delta_{i,j}H_i,
\]
where $C_{i,j}$ is the Cartan matrix of $\frak{g}$. Let $\Sigma^+$ denote a set of positive roots, and let $\Pi=\{\alpha_1,\ldots,\alpha_n\}$ be the set of simple roots. Note that  $X_i=X_{\alpha_i}$ and $Y_i=Y_{\alpha_i}=X_{-\alpha_i}$ and $C_{i,j}=\alpha_i(H_j)$. 
The negative root vectors $Y_{\alpha}$ are generated by
\[
[Y_\alpha, Y_{\beta}]=N_{-\alpha,-\beta}Y_{\alpha+\beta},\qquad\text{if}\quad \alpha+\beta\in\Sigma^+,
\]
where $N_{-\alpha,-\beta}$ are constants. 

The Lie algebra $\frak{g}$ admits the decomposition,
\[
\frak{g}=\frak{n}_-\oplus \frak{h}\oplus\frak{n}_+=\frak{n}_-\oplus\frak{b}_+,
\]
where $\frak{h}$ is a Cartan subalgebra, $\frak{n}_\pm$ are nilpotent subalgebras defined by $\frak{n}_{\pm}=\oplus_{\alpha\in\Sigma^+}\mathbb{C}X_{\pm\alpha}$, and $\frak{b}_+ = \frak{n}_+\oplus\frak{h}$ is
the Borel subalgebra.   We also fix a split Cartan subgroup $\mathcal{H}$ with $\text{Lie}(\mathcal{H}) = \mathfrak{h}$ and a Borel subgroup $\mathcal{B}$ with $\text{Lie}(\mathcal{B}) = \frak{b}_+$ with $\mathcal{B} = \mathcal{H}\mathcal{N}$ where $\mathcal{N}$ is a Lie group having the Lie algebra $\mathfrak{n}_+$. We also denote Lie groups $\mathcal{B}^-$ and $\mathcal{N}^-$  with $\text{Lie}(\mathcal{B}^-) = \mathfrak{b}_-$ and  $\text{Lie}(\mathcal{N}^-) = \mathfrak{n}_-$.
Then the full Kostant-Toda (f-KT) hierarchy on the Lie algebra $\mathfrak{g}$ is defined as follows.
Let $L_{\frak{g}}$ be the element defined by
\begin{align}\label{eq:Lax}
L_{\frak{g}}&=\sum_{i=1}^n a_i(\t)H_i+\sum_{i=1}^nX_i+\sum_{\alpha\in\Sigma^+}b_{\alpha}(\t)Y_\alpha,\\
&=\sum_{i=1}^n \left(a_i(\t)H_i+X_i+b_i(\t)Y_i\right)+\sum_{\alpha\in\Sigma^+\setminus \Pi}b_{\alpha}(\t)Y_\alpha,\nonumber
\end{align}
which is called the Lax matrix.
Here  $a_i(\t)$ and $b_{\alpha}(\t)$ are functions of the multi-time variables $\t=(t_{m_k}:k=1,2,\ldots,n)$
where $m_k$ is defined below. 
For each time variable, we have the generalized Toda hierarchy defined by
\begin{equation}\label{eq:fKT}
\frac{\partial L_{\frak{g}}}{\partial t_{m_k}}=[P_k, L_{\frak{g}}],\qquad\text{with}\quad P_k=\Pi_{\frak{b}_+}\nabla I_{k+1},\end{equation}
where  $\nabla$ is the gradient with respect to the Killing form, i.e. for any $x\in \frak{g}, dI_k(x)=K(\nabla I_k,x)$,
and $\Pi_{\mathfrak{b}_+}$ represents the $\mathfrak{b}_+$-projection.
Here the functions $I_k=I_k(L_{\frak{g}})$ are the Chevalley invariants which, for example in type $A$ are defined by
\[
I_{k+1}(L_{\frak{g}})=\frac{1}{k+1}\text{tr}(L_{\frak{g}}^{k+1}),\quad\text{which gives}\quad \nabla I_{k+1}=L_{\frak{g}}^k.
\]
\begin{remark}\label{rem:Chevalley}
In the case of the Lie algebra $\mathfrak{g}$ of type $D$, the recent paper \cite{KO:22} shows that one of the Chevalley invariants is given by a Pfaffian of the ring of polynomials $\C[\mathfrak{g}]$.
\end{remark}

Note here that $I_k$ is a weighted homogeneous polynomial of $(a_i,b_{\alpha})$ of weight $k$, where the weights are defined as
\[
\text{wt}(a_i)=1,\qquad \text{wt}(b_\alpha)=\text{ht}(\alpha)+1.
\]
Here $\text{ht}(\alpha)$ represents the height of $\alpha$, which is defined by $\sum_{i=1}^nc_i$ for $\alpha=\sum_{i=1}^nc_i\alpha_i$.
The index $m_k$ of the time variable in \eqref{eq:fKT} is then defined by $m_k=\text{wt}(I_{k+1})-1$.  
For example, in the case of $k=1$, the f-KT lattice gives the following set of differential equations for the variables $\{(a_i,b_\alpha):1\le i\le n, \alpha\in\Sigma^+\}$,
\begin{equation}\label{eq:fKT1}
\frac{\partial a_i}{\partial t_1}=b_i\quad (1\le i\le n),\qquad \frac{\partial b_\alpha}{\partial t_1}=-\sum_{i=1}^n (\alpha(H_i)a_i)\,b_\alpha+\sum_{i=1}^nN_{\alpha_i,-\alpha-\alpha_i}b_{\alpha+\alpha_i},
\end{equation}
where  $b_i=b_{\alpha_i}$, $[H_i,Y_{\alpha}]=-\alpha(H_i)Y_{\alpha}$ and $[X_{\alpha_i}, Y_{\alpha+\alpha_i}]=N_{\alpha_i, -\alpha-\alpha_i}Y_{\alpha}$ for $\alpha+\alpha_i\in \Sigma^+$.   Notice that wt$(\frac{\partial}{\partial t_1})=1$, and the f-KT lattice consists of the weighted homogeneous equations. 
Then one can see  that the solutions $b_{\alpha}(\t)$ of weight $\text{ht}(\alpha)+1$ are determined by the functions $a_i(\t)$ of weight $1$. The Kostant-Toda lattice hierarchy is then a special case with all $b_{\alpha}(\t)=0$ for $\alpha\in\Sigma^+\setminus\Pi$, 
that is, the Lax matrix $L_{\mathfrak{g}}$ is a Jacobi element of $\mathfrak{g}$. The hierarchy of this type is also referred to as the tri-diagonal KT hierarchy.
In this case, Flaschka and Heine in \cite{FH} introduced the so-called $\tau$-functions which express the variables $a_i(\t)$ by
\begin{equation}\label{def:tau}
a_i(\t)=\frac{\partial}{\partial t_1}\ln\tau_i(\t),\qquad (1\le i\le n),
\end{equation}
(see also \cite{Kos:79, GW:84}).
Then $b_i(\t)$ are expressed by
\[
b_i(\t)=\frac{\partial}{\partial t_1}a_i(\t)=b_i^0\prod_{j=1}^n\left(\tau_j(\t)\right)^{-C_{i,j}},
\]
where $C_{i,j}$ is the Cartan matrix and $b_i^0$  are constants. 

In general, the $\tau$-function is defined as follows (c.f. \cite{KO:22}). Let $V(\lambda)$ be the irreducible highest weight representation of 
the group $\mathcal{G}$ with highest weight $\lambda$ and the highest weight vector $v_{\lambda}$. Then
there exists a unique bilinear form $\langle\cdot,\cdot\rangle: V(\lambda)\times V(\lambda)\to \C$ with the normalization
$\langle v_{\lambda},v_\lambda\rangle_\lambda=1$, and a matrix coefficient $c_\lambda(g)$ for the $\mathcal{G}$-orbit of the highest weight vector $g\cdot v_\lambda$ for $g\in\mathcal{G}$ is given by
\begin{equation}\label{def:coefficient}
c_\lambda(g)=\langle v_\lambda, g\cdot v_\lambda\rangle_\lambda.
\end{equation}
Then the $\tau$-functions are defined as the matrix coefficient for $g(\t)=\exp(\Theta_{L^0}(\t))$ with the fundamental weight 
$\varpi_i$ (c.f. \cite{Symes, GW:84, KO:22}),
\begin{equation}\label{def:tauG}
\tau_i(\t):=c_{\varpi_i}(\t)=\langle v_{\varpi_i},g(\t)\cdot v_{\varpi_i}\rangle_{\varpi_i}.
\end{equation}

The  following lemma for the matrix coefficient is useful (the proof can be found in \cite{KO:22}).

\begin{lemma}
\label{lem:mat_coeff}
\begin{enumerate}
\item[(1)]
If $\{ u_1, \dots, u_r \}$ is an orthonormal basis of $V(\lambda)$, 
then we have for $v$, $w \in V(\lambda)$,
\begin{equation}
\label{eq:mat_coeff1}
\langle v, gh \cdot w \rangle_\lambda
 =
\sum_{i=1}^r
 \langle v, g \cdot u_i \rangle_\lambda
 \langle u_i, h \cdot w \rangle_\lambda.
\end{equation}
\item[(2)]
If $g \in GL_{n+1}$ is written as $g = {n}^- h n^+$ 
with ${n}^- \in {\mathcal{N}}^-$, $h \in \HH$ and $n^+ \in \mathcal{N}^+$, 
then we have
\begin{equation}
\label{eq:mat_coeff2}
c_\lambda({n}^- h n^+) = \chi^\lambda(h),
\end{equation}
where $\chi^\lambda: \HH \to \C^\times$ is the character 
corresponding to $\lambda$.
\item[(3)]
For dominant weights $\lambda$ and $\mu$, we have
\begin{equation}
\label{eq:mat_coeff3}
c_\lambda(g) \cdot c_\mu(g) = c_{\lambda+\mu}(g)
\qquad(g \in GL).
\end{equation}
\end{enumerate}
\end{lemma}

We now show that the solutions of the f-KT hierarchy \eqref{eq:fKT} are described 
in terms of the $\tau$-functions given by \eqref{def:tauG}. We also give a remark on the banded f-KT hierarchy.
First we note that the coefficients $b_\alpha$ of a solution 
to \eqref{eq:fKT} are uniquely determined by $a_1, \dots, a_n$.

\begin{proposition}
\label{prop:sol_b}
If $L = L(\t)$ is a solution of the form \eqref{eq:Lax} 
to the full Kostant--Toda hierarchy \eqref{eq:fKT}, 
then $b_\alpha$'s ($\alpha \in \Sigma^+$) 
are expressed as polynomials in $a_i$'s and their derivatives.
In particular, we have
\begin{equation}
\label{eq:b=a'}
b_{\alpha_i} = \frac{\partial a_i}{\partial t_1}
\qquad(1 \le i \le n).
\end{equation}
\end{proposition}

\begin{proof}
Recall that the height $\text{ht}(\alpha)$ of a root 
$\alpha = \sum_{i=1}^n c_i \alpha_i \in \Sigma$ is defined by 
$\text{ht}(\alpha) = \sum_{i=1}^n c_i$.
For a nonzero integer $k$, let $\g_k$ be the span of root vectors 
$X_\alpha$ with $\height(\alpha) = k$.
And we put $\g_0 = \h$.
Then we have
\[
\g = \bigoplus_{k \in \Int} \g_k,
\quad
[\g_k, \g_l] \subset \g_{k+l}.
$$
If we put
$$
L_1 = \sum_{i=1}^n X_{\alpha_i},
\quad
L_0 = \sum_{i=1}^n a_i(t) H_i,
\quad
L_{-k} = \sum_{\height(\alpha) = k} b_\alpha(t) X_{-\alpha}
\qquad(k > 0),
\]
then we have
\[
L = L_1 + L_0 + L_{-1} + L_{-2} + \cdots,
\quad
L^{\ge 0} = L_1 + L_0,
\quad
L_k \in \g_k.
\]
By comparing the graded components of \eqref{eq:fKT} with $m_k=1$, we obtain
\begin{gather}
\label{eq:deg0}
\frac{ \partial L_0 }{\partial t_1 } = [L_1, L_{-1}],
\\
\label{eq:degk}
\frac{ \partial L_{-k} }{\partial t_1 } = [L_1,L_{-k+1}] + [L_0, L_{-k}]
\quad(k>0).
\end{gather}

Since $[L_1, L_{-1}] = \sum_{i=1}^n b_{\alpha_i} H_i$, we obtain \eqref{eq:b=a'} 
from \eqref{eq:deg0} by comparing the coefficients of $H_i$.

We prove by induction on $\height(\beta)$ that 
$b_\beta$ can be uniquely expressed as a linear combination 
of $a_i b_\alpha$ and $\frac{\partial b_\alpha}{\partial t_1}$ with 
$1 \le i \le n$ and $\height(\alpha) = \height(\beta)-1$.
Suppose $k>0$.
Equating the coefficients of $X_{-\alpha}$ with $\height(\alpha) = k$ 
in \eqref{eq:degk}, we obtain
\begin{equation}
\label{eq:lin_eq}
\frac{\partial b_\alpha}{\partial t_1}
 =
- \left(\sum_{i=1}^n\alpha(H_i) a_i\right) b_\alpha 
+ \sum_{\height(\beta)=k+1} N_{\beta-\alpha,-\beta} b_\beta
\qquad(\height(\alpha) = k),
\end{equation}
where $N_{\beta-\alpha,-\beta} = 0$ unless $\beta-\alpha$ is a simple root.
We regard \eqref{eq:lin_eq} as a system of linear equations in the unknown variables 
$b_\beta$ ($\height(\beta)=k+1$) and consider the coefficient matrix
$$
M
 = 
\left( N_{\beta-\alpha,-\beta} \right)_{\height(\alpha) = k, \, \height(\beta)=k+1}.
$$
On the other hand, it can be shown that $M$ is the representation matrix 
of the linear map $\ad L_1 : \g_{-(k+1)} \to \g_{-k}$.
Since $L_1 = \sum_{i=1}^n X_{\alpha_i}$, 
we can find a principal $\spl_2$-triple $\{ h, e, f \}$ 
with $e = L_1$ such that $[h,e] = 2e$, $[h,f] = -2f$ and $[e,f] = h$, 
and we have $\g_k = \{ X \in \g : [h,X] = 2k X \}$ (see \cite[Section~5]{Kos:59}).
By appealing to the representation theory of $\spl_2$, 
we see that $\ad L_1 : \g_{-(k+1)} \to  \g_{-k}$ is injective.
Hence the matrix $M$ has a full rank 
and the system \eqref{eq:lin_eq} of linear equations has a unique solution 
in $(b_\beta)_{\height(\beta) = k+1}$.
\end{proof}

We then define the $l$-KT hierarchy on $\mathfrak{g}$ by imposing the following constraints (or reduction, see \cite[Section~5]{KY}),
\begin{equation}\label{eq:l-band}
b_\alpha = 0
\quad\text{with}\quad \height\, (\alpha) \ge \ell+1.
\end{equation}
Thus, the $\ell$-KT hierarchy is a reduction of the f-KT hiersrchy based on the root space decomposition with particular height, which we refer to as a \emph{root space reduction}.

\subsection{{Coordinates for the flag varieties}}
The $\tau$-functions near the Painlev\'e divisors can be expressed as follows. Assume that the matrix $g(\t)$ at a Painlev\'e
divisor $\t=\t_*$ has the Bruhat decomposition,
\begin{equation}
g(\t_*) = n_*\dot{w}_*b_* \qquad \text{for some} \qquad w_* \in \mathfrak{W},
\end{equation}
where $\mathfrak{W}$ is the Weyl group of $\mathcal{G}$. Let $\{u_1, \dots, u_r\}$ be an orthonormal basis of $V(\varpi_i)$, and set $\t \to \t + \t_*$, then by \eqref{eq:mat_coeff1} we have
\begin{align}\label{eq:tauexpansion}
\tau_i(\t) & = \langle v_{\varpi_i},g(\t) n_*\dot{w}_*b_*\cdot v_{\varpi_i}\rangle_{\varpi_i} \nonumber\\
& = \sum \limits_{i = 1}^r \langle v_{\varpi_i}, g(\t) \cdot u_i\rangle_{\varpi_i}\langle u_i, n_*\dot{w}_*b_* \cdot v_{\varpi_i}\rangle_{\varpi_i} \\
& = d_i \sum \limits_{i = 1}^r \langle v_{\varpi_i}, g(\t) \cdot u_i\rangle_{\varpi_i}\langle u_i, n_*\dot{w}_* \cdot v_{\varpi_i}\rangle_{\varpi_i}, \nonumber
\end{align}
where $d_i$ defined by $b_* \cdot v_{\varpi_i} = d_iv_{\varpi}$ is an overall non-essential constant which we set it to be $1$ in the following. Note that $\langle u_i, n_*\dot{w}_* \cdot v_{\varpi_i}\rangle_{\varpi_i}$ is certain matrix coefficient of a point in the flag variety $\mathcal{G} \slash \mathcal{B}$, and we recall some useful construction of coordinates on the flag varieties in the following.

The datum $\{\mathcal{H}, \mathcal{B}, \mathcal{B}^-, X_i, Y_i\,(1\le i\le n)\}$ is called a pinning in \cite{Lusztig3}. It is known that pinning is unique up to conjugation under $\mathcal{G}$. For $\alpha_i \in \Pi$, let $x_{i}: \mathbb{C} \to \mathcal{G}$ and $y_{i}: \mathbb{C} \to \mathcal{G}$ be the one-parameter subgroups given by
\begin{equation}\label{eq:xy}
x_{i}(\xi) = \exp (\xi X_i), \qquad y_{i}(\xi) = \exp(\xi Y_i), \quad \xi \in \mathbb{C}.
\end{equation}
For $\alpha_i \in \Pi$, a choice of representative of $s_i \in \mathcal{W}$ in $\mathcal{G}$ can be taken as
\[\dot{s}_{i} = x_{i}(-1)y_{i}(1)x_{i}(-1).\]

For $w \in \mathcal{W}$, the $\mathcal{B}$-orbit $\mathcal{O}^+_w = \mathcal{B}w\mathcal{B} \slash \mathcal{B}$ and $\mathcal{B}^-$-orbit $\mathcal{O}^-_w = \mathcal{B}^-w\mathcal{B} \slash \mathcal{B}$ in $\mathcal{G} \slash \mathcal{B}$ are called the Bruhat cell and opposite Bruhat cell corresponding to $w$, respectively. Let
\begin{equation*}
\mathcal{N}_w = \mathcal{N} \cap \dot{w}\mathcal{N}^-\dot{w}^{-1} \qquad \text{and} \qquad \mathcal{N}^-_w = \mathcal{N}^- \cap \dot{w} \mathcal{N}^-\dot{w}^{-1},
\end{equation*}
then we have the following isomorphisms
\begin{align*}
& \mathcal{N}_w \to \mathcal{B}w\mathcal{B} \slash \mathcal{B}, \qquad a \mapsto a\dot{w} \mathcal{B}, \quad a \in \mathcal{N}_w,\\
& \mathcal{N}^-_w \to \mathcal{B}^-w\mathcal{B} \slash \mathcal{B}, \qquad b \mapsto b\dot{w} \mathcal{B}, \quad b \in \mathcal{N}_w.
\end{align*}
Let $\ell(w)$ denote the length of $w \in \mathcal{W}$. Note that if $w = w_1w_2$ and $\ell(w) = \ell(w_1) + \ell(w_2)$, then
\begin{equation*}
\mathcal{N}_w\dot{w} = (\mathcal{N}_{w_1}\dot{w}_1)(\mathcal{N}_{w_2}\dot{w}_2)
\end{equation*}
is a direct product decomposition. Let $w = s_{i_1}s_{i_2} \cdots s_{i_k}$ be a reduced expression for $w$, then
\[\begin{array}{rccc}
\phi_w: & \mathbb{C}^k & \longrightarrow & \mathcal{O}^+_w\\
& (\xi_1, \dots, \xi_k) & \mapsto & x_{{i_1}}(\xi_1)\dot{s}_{{i_1}} \cdots x_{{i_k}}(\xi_k)\dot{s}_{i_k} \mathcal{B}
\end{array}\]
is an isomorphism (c.f. \cite{MR, LY}), and $(\xi_1, \dots, \xi_k)$ is called the Bott-Samelson coordinate on $\mathcal{O}^+_w$ defined by the reduced word $w = s_{i_1}s_{i_2} \cdots s_{i_k}$.

Let $w_0 \in \mathcal{W}$ be the longest Weyl group element. For $w \in \mathcal{W}$, let $\hat{w}$ be defined by
\[\hat{w} \cdot w = w_0 \qquad \text{or} \qquad \hat{w} = w_0 \cdot w^{-1}.\]
The identity
\begin{equation*}
\mathcal{N}^-_w = \dot{\hat{w}}^{-1}\mathcal{N}_{\hat{w}}\dot{\hat{w}}
\end{equation*}
gives rise to the isomorphism
\[\begin{array}{rccc}
\psi_w: & \mathcal{B}^- w \mathcal{B} \slash \mathcal{B} & \longrightarrow & \mathcal{O}^+_{\hat{w}}\\
& a \dot{w}\mathcal{B} & \mapsto & \dot{\hat{w}} a^{-1}\mathcal{B},
\end{array}\]
where $a \in \mathcal{N}^-_w$. A parameterization for $\mathcal{O}^-_w$ can thus be obtained as $\phi_{\hat{w}} \circ \psi^{-1}_{w}: \mathbb{C}^k \to \mathcal{O}^-_w$, where $k = \ell(\hat{w}) = \ell(w_0) - \ell(w)$, and following \cite{LY} we call it the Bott-Samelson coordinate for $\mathcal{O}^-_w$ associated with the reduced word $\hat{w} = s_{i_1}s_{i_2} \cdots s_{i_k}$.

Another coordinate which is defined only on an open subset of $\mathcal{O}^-_w$ can be described as follows. Recall that the Richardson variety $\mathcal{R}_{v, w}$ is defined as
\[\mathcal{R}_{v, w} := \mathcal{O}^+_w \cap \mathcal{O}^-_v,\]
which is nonempty precisely when $v \le w$ in the Bruhat order. Let $\hat{w}$ be defined as above, and $\hat{w} = s_{i_1}s_{i_2} \cdots s_{i_k}$ be a reduced expression for $\hat{w}$. Then the map
\begin{equation}\label{eq:Lusztig1}
\begin{array}{rccc}
\varphi_{\hat{w}}: & \mathbb{C}^k & \longrightarrow & \mathcal{R}_{w, w_0}\\
& (\xi_1, \dots, \xi_k) & \mapsto & y_{i_1}(\xi_1) y_{i_2}(\xi_2) \cdots y_{i_k}(\xi_k)\dot{w} \mathcal{B}
\end{array}
\end{equation}
defines an isomorphism (c.f. \cite{MR, KW:15}). We call the coordinate system on $\mathcal{R}_{w, w_0} \subset \mathcal{O}^-_w$  Lusztig coordinate for $\mathcal{R}_{w, w_0}$ associated with the reduced word $\hat{w} = s_{i_1}s_{i_2} \cdots s_{i_k}$. Note that the image of $\varphi_{\hat{w}}$ does not depend on the reduced word we choose for $\hat{w}$. Note also that $\mathcal{R}_{w, w_0} = w_0\mathcal{B}^-\mathcal{B} \cap \mathcal{O}^-_w$ is the intersection of $\mathcal{O}^-_w$ with the shifted big cell $w_0\mathcal{B}^-\mathcal{B}$, so that $\varphi$ only defines a coordinate on an open dense subspace of $\mathcal{O}^-_w$.

For simplicity we mainly use the Lusztig coordinate for a generic element in the corresponding Bruhat cell in the rest of the paper. We will give a concrete example to illustrate the difference between the two coordinate systems (see \eqref{eq:Bott-Samelson} and \eqref{eq:Lusztig} below for the parameterizations of a generic element in the flag variety of $B_2$).

We will discuss the f-KT hierarchies on the Lie algebras of $A$-, $B$- and $G$-types in this paper by embedding those
equations into the f-KT lattice in the Hessenberg form \eqref{Lax}.


\section{The f-KT hierarchy of type $A$}\label{sec:AKT}
Now let us consider the f-KT hierarchy discussed in the previous sections in the frame of the Lie algebra of type $A$,
\[
\mathfrak{sl}_{n+1}(\C):=\{X\in \text{End}(\C^{n+1}): \text{tr}(X)=0\}.
\]
Recall that $\mathfrak{g}=\mathfrak{sl}_{n+1}(\C)$ has a triangular decomposition, $\mathfrak{g}=\mathfrak{n}_+\oplus\mathfrak{h}\oplus\mathfrak{n}_-$, where $\mathfrak{h}$ (resp. $\mathfrak{n}_+, \mathfrak{n}_-$) is a subgroup of $\mathfrak{g}$ consisting of diagonal matrices (resp. strictly upper triangular matrices, strictly lower triangular matrices). Then the Cartan subalgebra $\mathfrak{h}$ is given by
\[
\mathfrak{h}=\left\{ \text{diag}(h_1,\ldots,h_{n+1}): ~\sum_{i=1}^{n+1}h_i=0,~h_i\in\C\right\}.
\]
Let $\varepsilon_i$ be the linear map $\ve_i:\mathfrak{h}\to\C$ given by $\varepsilon_i(\text{diag}(h_1,\ldots,h_{n+1}))=h_i$. Then 
the root system of $\mathfrak{g}$ with respect to $\mathfrak{h}$ is given by
\[
\Sigma_{A_n}=\left\{\pm(\ve_i-\ve_j): 1\le i<j\le n + 1\right\}.
\]
In particular, the positive root system is given by
\[
\Sigma^+_{A_n}=\left\{\varepsilon_i-\varepsilon_j= \alpha_i+\alpha_{i+1}+\cdots+\alpha_{j-1}: 1\le i< j\le n+1\right\}.
\]
Note that the simple roots are given by $\alpha_i=\varepsilon_i-\varepsilon_{i+1}$ and  $|\Sigma^+_{A_n}|=\frac{n(n+1)}{2}$.

We take the Chevalley basis in the form,
\[
H_i=E_{i,i}-E_{i+1,i+1}, \qquad X_i=E_{i,i+1},\qquad Y_i=E_{i+1,i},\qquad (1\le i\le n).
\]
The root vectors $Y_{\ve_i-\ve_j}, (i<j)$ are defined by
\[
Y_{\ve_i-\ve_j}=E_{j,i}=(-1)^{j-i}[[\cdots[[Y_i,Y_{i+1}],Y_{i+2}]\cdots ],Y_{j-1}].
\]
Then the Lax matrix $L_A\in\mathfrak{sl}_{n+1}(\C)$ for $\mathfrak{sl}_{n+1}(\C)$ is given by
\begin{align*}
L_A&=\sum_{i=1}^n(a_iH_i+X_i+b_iY_i)+\sum_{j=3}^{n+1}\sum_{i=1}^{j-1}b_{\ve_i-\ve_j}Y_{\ve_i-\ve_j}\\[1.0ex]
&=\begin{pmatrix}
a_1  & 1 & 0 & \cdots  & 0&0\\
b_1 & a_2-a_1 & 1 & \cdots &0 &0\\
b_{\ve_1-\ve_3}&b_2 & a_3-a_2& \ddots &\vdots &\vdots\\
\vdots &\vdots &\ddots &\ddots &1 &\vdots\\
b_{\ve_1-\ve_{n}}&b_{\ve_2-\ve_n}&\cdots &\cdots & a_n-a_{n-1}&1\\
b_{\ve_1-\ve_{n+1}}&b_{\ve_2-\ve_{n+1}}&\cdots&\cdots & b_n & -a_n
\end{pmatrix},
\end{align*}
with $b_i=b_{\ve_i-\ve_{i+1}}$, and the f-KT equation of $A$-type,
\[
\frac{\partial L_A}{\partial t_1}=[P_A, L_A]\qquad\text{with}\qquad P_A=\sum_{i=1}^n(a_iH_i+X_i),
\]
which gives
\[\left\{
\begin{array}{lll}
\displaystyle{\frac{\partial a_i}{\partial t_1}=b_i\qquad (1\le i\le n)},\\[2.0ex]
\displaystyle{\frac{\partial b_{\ve_i-\ve_j}}{\partial t_1}=-\left(\sum_{k=1}^n(\ve_i-\ve_j)(H_k)a_k\right) b_{\ve_i-\ve_j}+b_{\ve_i-\ve_{j+1}}-b_{\ve_{i-1}-\ve_j}\qquad (1\le i< j\le n)},
\end{array}
\right.
\]
where $\ve_i(H_k)=\delta_{i,k}-\delta_{i,k+1}$.
Comparing the Lax matrix $L_A$ with $L$ in \eqref{L}, we have
\begin{align*}
a_{i,i}&=a_i-a_{i-1}\qquad (1\le i\le n+1),\\
a_{\ell+k,k}&=b_{\ve_k-\ve_{\ell+k}}\qquad (1\le k\le n-\ell+1, \,1\le \ell\le n).
\end{align*}
These entries of $L_A$ can be written by
\begin{align*}
a_i&=\frac{\partial}{\partial t_1}\ln \tau^A_i(\t),\qquad (1\le i\le n),\\
b_{\ve_k-\ve_{\ell+k}}&=\frac{p_{\ell+1}(\tilde D)\tau^A_k(\t)\circ\tau^A_{\ell+k-1}(\t)}{\tau^A_k(\t)\tau^A_{\ell+k-1}(\t)}\qquad (1\le k\le n-\ell+1),
\end{align*}
where the $\tau^A$-functions are defined in \eqref{def:tauG}. 
Recall first that the exterior product space $\wedge^i\C^{n+1}$ gives the irreducible highest weight representation of $SL_{n+1}(\C)$ with the
highest weight $\varpi_i=\ve_1+\cdots+\ve_i$ and the highest weight vector $v_{\varpi_i}=e_1\wedge\cdots\wedge e_i$.
Note here that $\varpi_i(H_j)=\delta_{i,j}$. 
Then the $\tau$-functions are expressed as the determinants of the leading principal minors of $g(\t)=\exp(\Theta_{L^0}(\t))$, i.e.
\begin{equation}\label{eq:tauA}
\tau^A_i(\t)=\langle v_{\varpi_i}, g(\t)\cdot v_{\varpi_i}\rangle_{\varpi_i}=\langle e_1\wedge\cdots\wedge e_i, \, e^{\Theta_{L^0}(\t)}\,e_1\wedge\cdots\wedge e_i\rangle=\Delta_i(\t)\qquad (1\le i\le n),
\end{equation}
where $\langle\cdot,\cdot\rangle$ is nothing but the usual inner product (determinant) on $\wedge^i\C^{n+1}$, and
$\Delta_i(\t)$ is the $i$-th principle minor of the matrix $g(\t)$.
They are the well-known $\tau$-functions for the f-KT hierarchy of type $A$ (c.f. \cite{KW:15}). One should note that these formulas have been given in \cite{GW:84, FH} for
the case of tridiagonal Toda lattice. Here we claim that the formula is still valid for the f-KT hierarchy.

We here also give an expansion form of the $\tau$-function about a Painlev\'e divisor $\t_*$ with
$e^{\Theta_{L^0}(\t_*)}=n_*w_*b_*$, 
\begin{align}\label{eq:exp}
\tau_i^A(\t)&=\Delta_i(\t)=\langle e_1\wedge\cdots\wedge e_i, e^{\Theta_{L^0}(\t)} n_*w_*b_*\cdot e_1\wedge\cdots\wedge e_i\rangle\\
&=\sum_{1\le j_1<\cdots<j_i\le n+1}\langle e_1\wedge\cdots\wedge e_i, e^{\Theta_{L^0}(\t)} e_{j_1}\wedge\cdots\wedge e_{j_i}\rangle\,\langle e_{j_1}\wedge\cdots\wedge e_{j_i}, n_*w_*b_*\cdot e_1\wedge\cdots\wedge e_i\rangle.\nonumber
\end{align}
This formula is due to \eqref{eq:mat_coeff1} in Lemma \ref{lem:mat_coeff}. Here note that the set $\{e_{j_1}\wedge\cdots\wedge e_{j_i}:1\le j_1<\cdots<j_i\le n+1\}$ forms the orthonormal basis of $\wedge^i\C^{n+1}$.
If the initial matrix $L^0$ is nilpotent, the $\tau$-functions become polynomials of $\t$. We study the polynomial solutions in the next section.


\subsection{Polynomial $\tau$-functions}
In \eqref{tauW}, if we set all the eigenvalues zero, i.e. $\kappa_i=0$ for all $i$, then we have
\begin{equation}\label{eq:Apoly}
\tau_i^A(\t; w_*)=\langle e_1\wedge\cdots\wedge e_i,\mathcal{P}_n^A(\t){n_*} \cdot e_{w_*(1)}\wedge\cdots\wedge e_{w_*(i)}\rangle, \qquad (1\le i\le n),
\end{equation}
where $\mathcal{P}_n^A(\t)$ is the $(n+1)\times (n+1)$ matrix given in \eqref{eq:polyM}, 
\begin{equation}\label{eq:P}
\mathcal{P}_n^A(\t)=\mathcal{P}_{n+1}(\t)=\begin{pmatrix}
1 & p_1(\t) &p_2(\t) & p_3(\t)& \cdots  &p_n(\t) \\
0 & 1 & p_1(\t) & p_2(\t)&\cdots \cdots &p_{n-1}(\t)\\
\vdots&\ddots &\ddots&\ddots &\cdots&\vdots\\
0& \cdots &0&1&\cdots&p_{n-i+1}(\t)\\
\vdots &\vdots &\vdots&\ddots &\ddots &\vdots\\
0 &\cdots&\cdots&\cdots &0&1
\end{pmatrix}.
\end{equation}

We note that the classification of the polynomial solutions for type $A$ can be obtained by listing all permutations $w_*\in\mathfrak{S}_{n+1}$.  We also need an element $n_*$ in \eqref{eq:nwb} for each $w_*$, which is constructed as follows (see also \cite{MR, KW3}).  The matrix $y^A_i(\xi)\in SL_{n+1}$ defined in \eqref{eq:xy} is given by
\[
y_i^A(\xi):=\exp\left(\xi Y_i\right)=
\begin{pmatrix}
1 & & & & & \\
   & \ddots &  & & & \\
 & &  1 & & & \\
 & & \xi& 1& & \\
 & & & & &\ddots &\\
 & & & & & & 1
 \end{pmatrix}\qquad (1\le i\le n),
\]
where $\xi$ is a parameter.  For each $w_*\in \mathfrak{S}_{n+1}$, 
let $\hat w$ be defined by
 \begin{equation}\label{eq:omega}
 \hat w\cdot w_* = w_0\qquad\text{or}\qquad \hat w=w_0\cdot w_*^{-1},
 \end{equation}
 where $ w_0$ is the longest element in $\mathfrak{S}_{n+1}$.
We write $\hat w$ in a reduced expression, 
 \[
 \hat w=s_{i_1}s_{i_2}\cdots s_{i_m}\quad \text{with}\quad m=\ell(\hat w)=\ell(w_0)-\ell(w_*),
 \]
  where $\ell(w)$ denotes the length of $w$.
 Then the $n_*$ associated with $w_*$ in the Lusztig coordinate \eqref{eq:Lusztig1} is given by 
\begin{equation}\label{eq:nA}
n_*=y^A_{i_1}(\xi_{1})\,y^A_{i_2}(\xi_{2})\,\cdots\, y^A_{i_m}(\xi_{m}),
\end{equation}
where $\{\xi_1,\ldots,\xi_m\}$ gives the set of free parameters.  That is, the element $n_*\dot w_*$ mod$(\mathcal{B}^+)$ gives a generic element in the corresponding Bruhat cell of the flag $SL_{n+1}/\mathcal{B}^+$, and the dimension of the cell is $m=\ell(\hat w)$.

As an example, consider
 $w_*=s_1s_2s_3s_4s_1s_3s_2$ for the f-KT hierarchy of $A_4$-type. Then we have
 $\hat w=s_2s_4s_3$, and $n_*=y_2(\xi_1)y_4(\xi_2)y_3(\xi_3)$ is given by
 \[
 n_*\dot w_*=\begin{pmatrix}
 0 & 0 & 0 & 0 &1\\
 0&   0& 1 & 0 &0\\
 1&   0& \xi_1& 0 & 0\\
 \xi_3 & 0 &0 & -1 &0\\
 \xi_2\xi_3&1&0&-\xi_2  &0
 \end{pmatrix}.
 \]
 There are three free parameters for the Bruhat cell $X_{w_*}=\mathcal{N}^-\dot  w_*\mathcal{B}^+/\mathcal{B}^+$, i.e.  dim$(X_{ w_*})=3$.
Noting that
\begin{equation}\label{eq:w4}
 w_*(1,2,3,4,5)=(3,5,2,4,1),
\end{equation}
we have the polynomial solutions of the $\tau$-functions in \eqref{eq:Apoly} for $w_*=s_1s_2s_3s_4s_1s_3s_2$, 
\begin{align*}
\tau^A_1(\t,w_*)&=S_{\young[2][3]}(\t)+\xi_3 S_{\young[3][3]}(\t)+\xi_2\xi_3S_{\young[4][3]}(\t), \\
\tau^A_2(\t,w_*)&=S_{\young[3,2][3]}(\t),\\
\tau^A_3(\t,w_*)&=S_{\young[2,1,1][3]}(\t)+\xi_3S_{\young[2,2,1][3]}(\t)+\xi_1\xi_3S_{\young[2,2,2][3]}(\t),\\
\tau^A_4(\t,w_*)&=S_{\young[1,1,1,1][3]}(\t).
\end{align*}

We can now find the elements of the matrix $L_A$. To compute $p_{\ell+1}(\tilde D)\tau_k\circ\tau_{k+\ell-1}$,
the following formula is useful (c.f. \cite{K:18}),
\begin{align}\label{eq:Id1}
p_m(\tilde D)\tau_\alpha\circ\tau_\beta&=\sum_{i+j=m}\left(p_i(\tilde\partial)\tau_\alpha\right)\left(p_j(-\tilde \partial)\tau_\beta\right)\\
&=\sum_{i+j=m}\left(S_{\mu_i}(\tilde\partial)\tau_\alpha\right)\left((-1)^jS_{\nu_j}(\tilde\partial)\tau_\beta\right),\nonumber
\end{align}
where  $\mu_i$ represents the Young diagram with $i$ horizontal boxes, and $\nu_j$ represents  $j$ vertical boxes.

 Then we find that 
\begin{align*}
p_5(\tilde D)S_{\young[2][3]}\circ S_{\young[1,1,1,1][3]}&=S_{\young[1][3]}-S_{\young[1][3]}=0,\\
p_5(\tilde D)S_{\young[3][3]}\circ S_{\young[1,1,1,1][3]}&=S_{\young[2][3]}-S_{\young[1][3]}^2+S_{\young[1,1][3]}=0,\\
p_5(\tilde D)S_{\young[4][3]}\circ S_{\young[1,1,1,1][3]}&=S_{\young[3][3]}-S_{\young[2][3]}S_{\young[1][3]}+S_{\young[1][3]}S_{\young[1,1][3]}-S_{\young[1,1,1][3]}=0.
\end{align*}
Note that they are some Pl\"ucker relations.
This implies that we have
\[
b_{\alpha[1,4]}(\t)=a_{5,1}(\t)=\frac{p_5(\tilde D)\tau_1\circ\tau_4}{\tau_1\tau_4}=0.
\]
This gives that $b_{\alpha[1,3]}=a_{4,1}$ and $b_{\alpha[2,4]}=a_{5,2}$ satisfying
\[
a_{4,1}(\t)=a_{4,1}^0\frac{\tau_4\tau_0}{\tau_3\tau_1}=\frac{p_4(\tilde D)\tau_1\circ\tau_3}{\tau_1\tau_3}\qquad \text{and}\qquad
a_{5,2}(\t)=a_{5,2}^0\frac{\tau_5\tau_1}{\tau_4\tau_2}=\frac{p_4(\tilde D)\tau_2\circ\tau_4}{\tau_2\tau_4},
\]
where $\tau_0=1$ and $\tau_5=-1$.
One can also find that
\[
p_4(\tilde D)\tau_1\circ\tau_3=-\xi_3(\xi_1+\xi_2+\xi_3)\tau_4,\qquad p_4(\tilde D)\tau_2\circ\tau_4=0,
\]
that is, we have
\[
a_{4,1}^0=-\xi_3(\xi_1+\xi_2+\xi_3)\qquad\text{and}\qquad a_{5,2}^0=0.
\]
This example shows that the $\tau$-functions corresponding to $ w_*$ in \eqref{eq:w4} gives a rational solution of the f-KT hierarchy of $A_4$-type with the entries $b_{\alpha[1,5]}=b_{\alpha[2,5]}=0$, which is a 3-banded KT hierarchy.  

It should be note that if $a_{4,1}^0=0$, then we have a solution of a 2-banded KT hierarchy, i.e. $a_{5,1}=a_{4,1}=a_{5,2}=0$. Then we can find that
\[
p_{3}(\tilde D)\tau_1\circ\tau_2=\tau_3,\qquad p_3(\tilde D)\tau_2\circ\tau_3=\tau_1\tau_4,\qquad
p_3(\tilde D)\tau_3\circ\tau_4=\tau_2\tau_5,
\]
and 
\[
a_{3,1}(\t)=\frac{\tau_3(\t)}{\tau_1(\t)\tau_2(\t)},\qquad a_{4,2}(\t)=\frac{\tau_1(\t)\tau_4(\t)}{\tau_2(\t)\tau_3(\t)},\qquad a_{5,3}=-\frac{\tau_2(\t)}{\tau_3(\t)\tau_4(\t)}.
\]



\section{The f-KT hierarchy of type $B$}\label{S:BKT}
We consider the Lie algebra of type $B$ in the following form,
\begin{equation}\label{eq:Btype}
\mathfrak{o}_{2n+1}(\C)=\{\,X\in\text{End}(\C^{2n+1}): X^TF+FX=0\,\},
\end{equation}
where the matrix $F$ is given by
\begin{equation}\label{eq:F}
F=\sum_{i=1}^{2n+1}(-1)^{n-i+1}E_{i, 2n+2-i}=\begin{pmatrix}
0 &\cdots & 0   & 0 & 0 &  \cdots & (-1)^n   \\
\vdots &\udots & &\vdots & & \udots&\vdots\\
0& & 0 & 0 &-1 &  &0\\
0&\cdots & 0& 1 & 0&\cdots & 0\\
0&  &-1&0&0&&0\\
\vdots &\udots& &\vdots & &\udots &\vdots\\
(-1)^n &\cdots&0 &0&0&\cdots &0
\end{pmatrix}.
\end{equation}
The Chevalley vectors $\{H_i,X_i,Y_i\}$ satisfy the conditions,
\[
[H_i,X_{j}]=C_{j,i}X_{j},\qquad [H_i,Y_{j}]=-C_{j,i}Y_{j},\qquad [X_{i},Y_{j}]=\delta_{i,j}H_i,
\]
where $C_B=(C_{i,j})$ is the Cartan matrix of $B$-type,
\[
C_B=\begin{pmatrix}
2 & -1 & 0 & \cdots &0\\
-1&2&-1&\cdots &0\\
\vdots &\ddots &\ddots & \cdots&\vdots \\
0&\cdots &-1&2&-2 \\
0&\cdots &0&-1&2
\end{pmatrix}.
\]
We take the following representation for the Chevalley generators (c.f. \cite{KY}),
\begin{align*}
H_i&=E_{i,i}-E_{i+1,i+1}+E_{2n+1-i,2n+1-i}-E_{2n+2-i,2n+2-i},\qquad (1\le i\le n-1),\\
H_n&=2(E_{n,n}-E_{n+2,n+2}),
\end{align*}
the simple root vectors,
\[
 X_{i}=E_{i,i+1}+E_{2n+1-i,2n+2-i}\qquad  (1\le i\le  n),
\]
and the negative generators,
\begin{equation}\label{eq:BY}
\left\{\begin{array}{lll}
Y_{i}=E_{i+1,i}+E_{2n+2-i,2n+1-i} \qquad (1\le i\le n-1),\\[1.0ex]
Y_{n}=2(E_{n+1,n}+E_{n+2,n+1}).
\end{array}\right.
\end{equation}
Then the negative root vectors  $Y_{\alpha}$ are generated by taking the commutators,
\[
[Y_{\alpha}, Y_{\beta}]=N_{-\alpha,-\beta}\,Y_{\alpha+\beta},
\]
with $N_{-\alpha,-\beta}\ne 0$ if $\alpha+\beta\in\Sigma^+$. Let $\varepsilon_i$ be the basis vector of $\mathfrak{h}^*$ dual to $\hat{E}_i=E_{i,i}-E_{2n+2-i,2n+2-i}$, i.e. $\varepsilon_i(\hat{E}_j)=\delta_{i,j}$,
then the set of positive roots is given by
\[
\Sigma^+=\left\{\begin{array}{clll}
\varepsilon_i~ (1\le i\le n),\quad\varepsilon_i\pm \varepsilon_j ~(1\le i<j\le n)
\end{array}\right\},
\]
where $\varepsilon_i=\alpha_i+\cdots +\alpha_{n}$.
Note that $|\Sigma^+|=n^2$.

\subsection{The f-KT hierarchy on $\mathfrak{so}_{2n+1}$}

The Lax matrix $L_B$ of the f-KT hierarchy on $\mathfrak{o}_{2n+1}(\C)$ is defined by
\[
L_B=\sum_{i=1}^ n a_iH_i+\sum_{i=1}^ n X_{i}+\sum_{\alpha\in\Sigma^+}b_{\alpha}Y_{\alpha}.
\]
The total number of the dependent variables $\{a_i, b_\alpha\}$ is $n^2+n$.  The set of eigenvalues of $L_B$ is given by
\[
\{\,0,\,\, \pm\kappa_i\,\, (1\le i\le n)\,\},
\]
where we assume $\kappa_i\ne 0$ and $\kappa_i\ne \kappa_j\,(i\ne j)$.
Notice that the invariants $I_{m}=\frac{1}{m}\text{tr}(L_B^{m})$ vanish when $m=$ odd.
Then the f-KT hierarchy of $B$-type is defined only for odd times $\t_B=(t_1,t_3,t_5,\ldots)$, and
each invariant $I_{2k+2}$ defines the Lax equation,
\[
\frac{\partial L_B}{\partial t_{2k+1}}=[P_{2k+1},L_B],\qquad P_{2k+1}:=\Pi_{\mathfrak{b}^+}\nabla(I_{2k+2})=(L_B^{2k+1})_{\ge 0}.
\]

For example, the f-KT lattice of $B_3$-type is formulated as follows.
The positive root system $\Sigma_{B_3}^+$ is given by the set of $3^2$ elements, and the corresponding root vectors are
\begin{align*}
&Y_{\vep_1-\vep_2}=Y_1=E_{2,1}+E_{7,6},\quad Y_{\vep_2-\vep_3}=Y_2=E_{3,2}+E_{6,5},\quad Y_{\vep_3}=Y_3=2(E_{4,3}+E_{5,4}),\\
&Y_{\vep_1-\vep_3}=E_{3,1}-E_{7,5},\quad Y_{\vep_2}=E_{4,2}-E_{4,2},\quad Y_{\vep_2+\vep_3}=E_{5,2}+E_{6,3},\\
&Y_{\vep_1}=E_{4,1}+E_{7,4},\quad Y_{\vep_1+\vep_3}=E_{5,1}-E_{7,3},\quad Y_{\vep_1+\vep_2}=E_{6,1}+E_{7,2}.
\end{align*}
The Lax matrix is then given by
\begin{align*}
L_{B_3}&=\sum_{i=1}^3 (a_iH_i+X_{i}+b_iY_{i})+\sum_{\alpha\in\Sigma_{B_3}^+\setminus\Pi}b_{\alpha}Y_{\alpha}\\
&=\begin{pmatrix}
a_1 & 1 &  & & &  \\
b_1 &a_2-a_1 & 1 &  & & & \\
b_{\vep_1-\vep_3} & b_2 & 2a_3-a_2 & 1 & & &  \\
b_{\vep_1}& b_{\vep_2} &  b_3 & 0 & 1 & &  \\
b_{\vep_1+\vep_3} &b_{\vep_2+\vep_3} & 0 &b_3 & -2a_3+a_2 & 1 &  \\
b_{\vep_1+\vep_2} & 0 & b_{\vep_2+\vep_3} & -b_{\vep_2} & b_2 &-a_2+a_1 & 1\\
0 & b_{\vep_1+\vep_2} & -b_{\vep_1+\vep_3} & b_{\vep_1} & -b_{\vep_1-\vep_3} &b_1 & -a_1
\end{pmatrix}.
\end{align*}
The first member of the f-KT hierarchy of $B$-type in rank $3$ is then given by
\begin{align*}
&\frac{\partial a_i}{\partial t_1}=b_i,\quad (i=1,2,3),\qquad\frac{\partial b_1}{\partial t_1}=-\sum_{j=1}^3\left(C_{1,j}a_j\right)b_1+b_{\vep_1-\vep_3},\\
& \frac{\partial b_2}{\partial t_1}=-\sum_{j=1}^3\left(C_{2,j}a_j\right)b_2-b_{\vep_1-\vep_3}+b_{\vep_2},\quad  \frac{\partial b_3}{\partial t_1}=-\sum_{j=1}^3\left(C_{3,j}a_j\right)b_3-b_{\vep_2},\\
& \frac{\partial b_{\vep_1-\vep_3}}{\partial t_1}=-(a_1+a_2-2a_3)b_{\vep_1-\vep_3}+b_{\vep_1},\quad \frac{\partial b_{\vep_2}}{\partial t_1}=(a_1-a_2)b_{\vep_2}-b_{\vep_1}+b_{\vep_2+\vep_3},\\
& \frac{\partial b_{\vep_1}}{\partial t_1}=-a_1b_{\vep_1}+b_{\vep_1+\vep_3},\quad \frac{\partial b_{\vep_2+\vep_3}}{\partial t_1}=(a_1-2a_3)b_{\vep_2+\vep_3}-b_{\vep_1+\vep_3},\\
& \frac{\partial b_{\vep_1+\vep_3}}{\partial t_1}=-(a_1-a_2+2a_3)b_{\vep_1+\vep_3}+b_{\vep_1+\vep_2},\quad \frac{\partial b_{\vep_1+\vep_2}}{\partial t_1}=-a_2b_{\vep_1+\vep_2},
 \end{align*}
where $C_{i,j}=\alpha_i(H_j)$ is the Cartan matrix of $B_3$.

\subsection{The $\tau$-functions for the f-KT hierarchy of $B$-type}
Embedding the Lax matrix $L_B$ into that of Hessenberg form ($(2n+1)\times(2n+1)$ matrix), the diagonal elements are expressed as
\begin{equation}\label{eq:ai}
\left\{\begin{array}{lll}
a_{i,i}=a_i-a_{i-1}\quad (1\le i\le n-1),\qquad a_{n,n}=2a_n-a_{n-1}, \\[1.0ex]
a_{n+1,n+1}=0,\qquad a_{n+1+k,n+1+k}=-a_{n+1-k,n+1-k} \quad (1\le k\le n),
\end{array}\right.
\end{equation}
where we assume $a_0=0$.  The elements in the lower triangular part satisfy the following symmetry
\begin{equation}
\label{eq:a-symm}
a_{\ell+k,k}=-(-1)^{\ell} a_{2n+2-k,2n+2- \ell -k}\qquad\text{for}\quad 
\left\{\begin{array}{lll}
1\le k\le\displaystyle{\left\lfloor\frac{2n-\ell+1}{2}\right\rfloor},\\[2.0ex]
1\le \ell \le 2n-1.
\end{array}\right.
\end{equation}
Note here that if $\ell$ is even, say $\ell=2m$, then we have $a_{n+m+1,n-m+1}=0$, which is an anti-diagonal
element in $L_B$. 

Recall that $a_{i,j}(t) $ are given in Proposition \ref{prop:aij} with the determinants $\Delta_i(\t)$. Then, the elements $a_i(t)$'s in \eqref{eq:ai} can be expressed by
\begin{equation}\label{eq:a-tau}
\left\{\begin{array}{lll}
\displaystyle{a_i(\t_B)=\frac{\partial}{\partial t_1}\ln \Delta_i(\t_B)},\qquad (1\le i\le n-1),\\[2.0ex]
\displaystyle{a_n(\t_B)=\frac{1}{2}\frac{\partial }{\partial t_1}\ln \Delta_n(\t_B)},
\end{array}\right.
\end{equation}
where the $\Delta_i(\t_B)$ for $1\le i\le n$ are defined by
\[
\Delta_i(\t_B):=\langle e_1\wedge\cdots\wedge e_i, e^{\Theta_{L_B^0}(\t_B)}\cdot e_1\wedge\cdots\wedge e_i\rangle\qquad\text{with}\qquad \Theta_{L_B^0}(\t_B):=\sum_{k=0}^\infty (L_B^0)^{2k+1}t_{2k+1}.
\]
Here $L_B^0$ is the initial Lax matrix. Note that $\Theta_{L_B^0}\in\mathfrak{o}_{2n+1}(\C)$, so that $g(\t_B):=e^{\Theta_{L_B^0}(\t_B)}$ is an element of the group $O_{2n+1}(\C)$ with $\text{Lie}(O_{2n+1})=\mathfrak{o}_{2n+1}$, i.e. (see \eqref{eq:Btype})
\[
g(\t_B)^TFg(\t_B)=F\qquad \text{where}\qquad F=\sum_{i=1}^{2n+1}(-1)^{n-i+1}E_{i, 2n+2-i}.
\]
The group $O_{2n+1}(\C)$ is then defined by
\[
O_{2n+1}(\C):=\left\{X\in \text{Aut}(\C^{2n+1}): X^TFX=F\right\}.
\]
We then define the $\tau$-functions for the f-KT hierarchy of $B$-type by
\begin{equation}\label{eq:Btau}
\left\{\begin{array}{lll}
\displaystyle{\tau^B_i(\t_B)}=\Delta_i(\t_B)\qquad (1\le i\le n-1),\\[1ex]
\displaystyle{\tau^B_n(\t_B)=\sqrt{\Delta_n(\t_B)}},
\end{array}\right.
\end{equation}
so that the variables $a_i(\t_B)$ can be written in the form \eqref{def:tau},
\[
a_i(\t_B)=\frac{\partial}{\partial t_1}\ln\tau_i^B(\t_B) \qquad (1\le i\le n),
\]
which is introduced in \cite{FH,GW:84} for the tridiagonal cases.

\begin{remark}
The expression \eqref{eq:Btau} of the $\tau$-function $\tau_n^B(\t_B)$  is related to a \emph{spin} representation of $SO_{2n+1}(\C)$ as follows.
Recall that the fundamental weights of $\mathfrak{so}_{2n+1}(\C)$ are given by
\begin{align*}
\varpi_i=\varepsilon_1+\cdots +\varepsilon_i\quad (1\le i\le n-1),\qquad  \varpi_n=\frac{1}{2}(\varepsilon_1+\cdots +\varepsilon_n),
\end{align*}
where $\varepsilon_i$ is the basis vector of $\mathfrak{h}^*$ dual to $\hat{e}_i=E_{i,i}-E_{2n+2-i,2n+2-i}$, i.e. $\varepsilon_i(\hat{e}_j)=\delta_{i,j}$. 
This implies that the exterior product $\wedge^i\C^{2n+1}$ of the standard representation gives the irreducible representation of $\mathfrak{so}_{2n+1}(\C)$ with highest weight $\varpi_i=\varepsilon_1+\cdots+\varepsilon_i$ for $1\le i\le n$, and we have
\[
\wedge^i\C^{2n+1}\cong V(\varpi_i)\quad(1\le i\le n-1),\qquad \wedge^n\C^{2n+1}\cong V(2\varpi_n).
\]
We then denote by $c_\lambda^{\text{Spin}_{2n+1}}(\tilde g)$ and $c_\lambda^{SO_{2n+1}}(g)$ the matrix coefficients on the spin group
$\text{Spin}_{2n+1}(\C)$ and the orthogonal group $SO_{2n+1}(\C)$, where $g=\pi (\tilde g)\in SO_{2n+1}(\C)$ with the covering map $\pi:\text{Spin}_{2n+1}(\C)\to SO_{2n+1}(\C)$. Then, from \eqref{eq:mat_coeff3} in Lemma \ref{lem:mat_coeff}, we have
\begin{equation}\label{eq:Square}
\left(c_{\varpi_n}^{\text{Spin}_{2n+1}}(\tilde g)\right)^2=c_{2\varpi_n}^{\text{Spin}_{2n+1}}(\tilde g)=c_{2\varpi_n}^{SO_{2n+1}}(g)=\langle v_{2\varpi_n}, g\cdot v_{2\varpi_n}\rangle_{2\varpi_n}.
\end{equation}
Taking $g=g(\t_B):=\exp(\Theta_{L^0}(\t_B))$, we have 
\[
\tau_n^B(\t_B)=c_{\varpi_n}^{\text{Spin}_{2n+1}}(\tilde g(\t_B))\quad \text{with}\quad g(\t_B)=\pi(\tilde g(\t_B)).
\]
In particular, the representation $V(\varpi_n)$ is a \emph{spin} representation, and $\tau^B_n(\t_B)$ can be expressed by a \emph{Pfaffian}. It seems that the explicit formula of $c_{\varpi_i}^{\text{Spin}_{2n+1}}(\tilde g)$ in terms of a Pfaffian
may not be so simple. However the Pfaffian formula of the polynomial $\tau$-functions has been obtained in \cite{Xie:22a}.
\end{remark}

One should also note that the symmetry in \eqref{eq:ai} then implies that the determinants $\Delta_i(\t_B)$
satisfy the following relations,
\begin{equation}\label{eq:relationD}
\Delta_i(\t_B)=\pm \Delta_{2n+1-i}(\t_B)\qquad (1\le i\le n).
\end{equation}
(One can show that the Young diagrams corresponding to those minors are in conjugation, see Proposition \ref{duality}.)
 That is, the solution is determined by $\Delta_i(\t)$ only for $1\le i\le n$ in
\eqref{eq:DN},
\[
\Delta_i(\t_B)=d_i\langle e_1\wedge\cdots\wedge e_i, Ve^{\Theta_{K}(\t_B)}V^{-1}{n_0n_*}\dot w_* \cdot e_1\wedge\cdots \wedge e_i\rangle.
\]
Here the unipotent matrix $n_*\in\mathcal{N}^-$ is given as follows.
Let us first give the Weyl group $\frak{W}_{B_n}$ associated with the algebra of type $B_n$. 
Each generator of the Weyl group $\frak{W}_{B_n}$ can be expressed by the simple reflections of $A$-type through the Bruhat embedding \cite{BL},
\begin{equation}\label{WeylB}
s^B_i=\left\{\begin{array}{lll}
s^A_is^A_{2n+1-i},\qquad 1\le i\le n-1,\\[1.0ex]
s_n^As_{n+1}^As_n^A, \qquad i=n,
\end{array}\right.
\end{equation}
where $s_i^A$ is a simple reflection of the symmetric group $\frak{S}_{2n+1}$. We have
\[
\frak{W}_{B_n}=\left\langle s_i^B~\Big|~(s_i^B)^2= (s_i^Bs_j^B)^2=e~(|i-j|\ge2),\,(s_{i-1}^Bs_i^B)^3=e~(2\le i\le n-1),\, (s_{n-1}^Bs_n^B)^4=e\,\right\rangle, 
\]
and  $\mathfrak{W}_{B_n}\cong \frak{S}_n\ltimes (\mathbb{Z}_2)^n$, $|\mathfrak{W}_{B_n}|=2^n n!$.  The longest element $w_0$ can be expressed as
\[
w_0=(s_1^Bs_2^B\cdots s_n^B)\cdot(s_1^Bs_2^B\cdots s_n^B)\,\cdots\, (s_1^Bs_2^B\cdots s_n^B),\qquad\text{with}\qquad \ell(w_0)=n^2,
\]
where the elements in the parenthesis repeats $n$ times. 
For example, the $\mathfrak{W}_{B_2}$ is given by
\[
\mathfrak{W}_{B_2}=\left\{ e, ~s_1^B, ~s_2^B, ~s_1^Bs_2^B, ~s_2^Bs_1^B, ~s_1^Bs_2^Bs_1^B, ~s_2^Bs_1^Bs_2^B, s_1^Bs_2^Bs_1^Bs_2^B=w_0\right\}.
\]
In order to compute $n_*\dot w_*$, we first define $y^B_i(\xi)\in O_{2n+1}(\C)$ by
 \begin{equation}\label{eq:pinningB}
 y^B_i(\xi)=\exp(\xi Y_i),\qquad(1\le i\le n),
  \end{equation}
 where $Y_i$'s are defined in \eqref{eq:BY}.
 Also the representation $\dot s_i^B\in O_{2n+1}(\C)$ are defined by
 \begin{equation}\label{eq:pinningB}
 \left\{\begin{array}{lll}
 \dot s^B_i=\dot s_i^A\dot s_{2n+1-i}^A,\qquad(1\le i\le n-1),\\[2.0ex]
  \dot s^B_n=\dot s_n^A\dot s_{n+1}^A\dot s^A_n,
   \end{array}\right.
 \end{equation}
 where $s_i^A\in\mathfrak{S}_{2n+1}$.

Now, according to \eqref{eq:Lusztig1}, let $\hat w$ be defined by the following form with a reduced expression,
\[
\hat w=w_0\cdot w_*^{-1}=s^B_{i_1}s^B_{i_2}\cdots s^B_{i_m}.
\]
Then we define the unipotent matrix $n_*$ by
\[
n_*=y^B_{i_1}(\xi_1)\, y^B_{i_2}(\xi_2)\,\cdots\, y^B_{i_m}(\xi_m),
\]
and $n_*\dot  w_*$ mod$(\mathcal{B}^+)$ gives the generic point of the Bruhat cell $\mathcal{N}^-\dot w_*\mathcal{B}^+$ in the flag variety $O_{2n+1}(\C)/\mathcal{B}^+$.

We now compute the minor $\Delta_n(\t_B)$ for $ w_*\in \mathfrak{W}_B$. To do this, we consider the expansion (see \eqref{eq:exp}, also Lemma \ref{lem:mat_coeff}),
\begin{align}\label{eq:DeltaB}
\Delta_n(\t_B)&=\langle e_1\wedge\cdots\wedge e_n,~g(\t_B)n_* \dot w_*\cdot e_1\wedge\cdots \wedge e_n\rangle\\
&=\sum_{1\le i_1<\cdots<i_n\le 2n+1}\langle e_1\wedge\cdots\wedge e_n,~g(\t_B)\cdot e_{i_1}\wedge \cdots\wedge e_{i_n}\rangle\cdot \langle e_{i_1}\wedge\cdots\wedge e_{i_n},~n_*\dot w_*\cdot e_{1}\wedge \cdots\wedge e_{n}\rangle,\nonumber
\end{align}
where $g(\t_B)=e^{\Theta_{L_B^0}(\t_B)}$ and $g(\t_B^*)=n_* w_*b_*$ (note here that we set $d_n=1$ in \eqref{eq:DN}). In the expansion, we first note that for each $ w\in\mathfrak{W}_B=\mathfrak{S}_n\ltimes \mathbb{Z}_2^n$, there exists $r$ in the subgroup $\mathfrak{N}_B \cong \mathbb{Z}_2^n$ such that
\begin{equation}\label{eq:w-r}
\dot w\cdot e_1\wedge\cdots\wedge e_n=\pm\dot r\cdot e_1\wedge\cdots\wedge e_n,
\end{equation}
where  $\mathfrak{N}_B$ is the normal subgroup representing $\Z_2^n$, which may be defined by
\[
\mathfrak{N}_B=\left\langle r_i\in\mathfrak{W}_B~(1\le i\le n): r_i(j)=\left\{\begin{array}{cll}
2n+2-j,~~&\text{if}~ j=i,\\
j,\quad &\text{if}~ j\ne i.
\end{array}\right.\right\rangle.
\]
Note that  $\mathfrak{N}_B$ is an Abelian group, and $|\mathfrak{N}_B|=2^n=\sum_{k=0}^n\binom{n}{k}$.
Each element can be expressed by
\[
r_k=s^B_ks_{k+1}^B\cdots s_{n-1}^Bs_n^Bs_{n-1}^B\cdots s_{k+1}^Bs_k^B\qquad\text{for}\quad k=1,\ldots,n.
\]
For example, in the case of $B_2$ with $ w_*=e$, the matrix $n_*$ is given by $n_*=y^B_1(\xi_1)y^B_2(\xi_2)y^B_1(\xi_3)y^B_2(\xi_4)$ in the Lusztig coordinate,
\begin{equation}\label{eq:Lusztig}
n_*=\begin{pmatrix}
1 & 0 & 0 & 0 & 0 \\
\xi_1+\xi_3 & 1 & 0 & 0 & 0\\
2\xi_2\xi_3 & 2(\xi_2+\xi_4) & 1 & 0&0\\
2\xi_2^2\xi_3 & 2(\xi_2+\xi_4)^2 & 2(\xi_2+\xi_4) & 1 &0\\
2\xi_1\xi_2^2\xi_3& 2(\xi_3\xi_4^2+\xi_1(\xi_2+\xi_4)^2) & 2(\xi_3\xi_4+\xi_1\xi_2+\xi_1\xi_4) & \xi_1+\xi_3 & 1
\end{pmatrix}.
\end{equation}
We then have
\begin{align*}
&\langle e_1\wedge e_2, \,n_*\cdot e_1\wedge e_2\rangle=\langle e_1\wedge e_2, \, \dot s_1^Bn_*\cdot e_1\wedge e_2\rangle=1,\\
&\langle \dot{r}_2\cdot e_1\wedge e_2, \,n_*\cdot e_1\wedge e_2\rangle=\langle e_1\wedge e_2, \,\dot s_2^Bn_*\cdot e_1\wedge e_2\rangle=2(\xi_2+\xi_4)^2,\\
&\langle \dot{r}_1\cdot e_1\wedge e_2, \,  n_*\cdot e_1\wedge e_2\rangle=\langle e_1\wedge e_2, \, \dot s_2^B\dot s_1^B n_*\cdot e_1\wedge e_2\rangle=2(\xi_3\xi_4+\xi_1(\xi_2+\xi_4))^2,\\
&\langle \dot{r}_1\dot{r}_2\cdot e_1\wedge e_2, \,n_*\cdot e_1\wedge e_2\rangle=\langle e_1\wedge e_2, \,\dot s_2^B\dot s_1^B\dot s_2^Bn_*\cdot e_1\wedge e_2\rangle=4(\xi_2\xi_3\xi_4)^2.
\end{align*}

\begin{remark}
These polynomials $\{1,\, \xi_2+\xi_4, \, \xi_3\xi_4+\xi_1(\xi_2+\xi_4),\, \xi_2\xi_3\xi_4\}$ are the spinors of the spin representation of $n_*$. Notice that these coordinates are useful for the total positivity of the flag variety of $B$-type.
\end{remark}

\begin{remark}
In the Bott-Samelson coordinate, the matrix $n_*$ for the same $w_*$ is given by
\begin{equation}\label{eq:Bott-Samelson}
{n_*=\begin{pmatrix}
1 & 0 & 0 & 0 & 0 \\
\xi_1 & 1 & 0 & 0 & 0\\
-2\xi_2 + 2\xi_1\xi_4 & 2\xi_4 & 1 & 0&0\\
2(\xi_3 - 2\xi_2\xi_4 + \xi_1\xi_4^2) & 2\xi_4^2 & 2\xi_4 & 1 &0\\
-2\xi_2^2+ 2\xi_1\xi_3 & 2\xi_3 & 2\xi_2 & \xi_1 & 1
\end{pmatrix},}
\end{equation}
which may be useful when we discuss several limits $\xi_i\to 0$.
\end{remark}

Equation \eqref{eq:Square} implies that the terms associated with $r\in\mathfrak{N}_B$ are complete squares,
\begin{align*}
\langle e_1\wedge\cdots\wedge e_n,~g(\t_B)\dot{r}\cdot e_1\wedge\cdots \wedge e_n\rangle&=\left(\phi_r(\t_B)\right)^2,\\
\langle \dot{r} e_1\wedge\cdots\wedge e_n,~n_*\dot w_*\cdot e_1\wedge\cdots \wedge e_n\rangle&=\left(\eta_r(\xi; w_*)\right)^2,
\end{align*}
where $\phi_r(\t_B)$ is an exponential polynomial, and $\eta_r(\xi, w_*)$ is a homogeneous polynomial of $\xi_i$'s. Now we have the main theorem in this section.
\begin{theorem}\label{thm:pfaffian}
The $\tau$-function associated with $ w_*$ denoted by $\tau^B_n(\t_B; w_*)$ is given by
\[
\tau^B_n(\t_B, w_*)=\sqrt{\Delta_n(\t_B)}=\sum_{r\in\mathfrak{N}_B}\phi_r(\t_B)\eta_r(\xi; w_*).
\]
\end{theorem}
\begin{proof}
 From \eqref{eq:mat_coeff3} in Lemma \ref{lem:mat_coeff}, we first note that 
\[
\langle e_1\wedge\cdots\wedge e_n, g(\t_B)h\cdot e_1\wedge\cdots\wedge e_n\rangle
=\left(\langle v_{\varpi_n}, \tilde g(\t_B)\tilde h\cdot v_{\varpi_n}\rangle_{\varpi_n}\right)^2,
\]
where $\pi(\tilde g(\t_B))=g(\t_B)$ and $\pi(\tilde h)=n_*\dot{w}_*$ with the covering map $\pi:\text{Spin}_{2n+1}\to SO_{2n+1}$.  Also note that $\phi_r(\t_B)$ and $\eta_r(\xi;w_*)$ can be expressed by
\begin{align*}
\phi_r(\t_B)=\langle v_{\varpi_n},\tilde g(\t_B)\cdot u_r\rangle_{\varpi_n}\quad \text{and}\quad 
\eta_r(\xi;w_*)=\langle u_r,\tilde h\cdot v_{\varpi_n}\rangle_{\varpi_n},
\end{align*}
where $\{u_r:r\in\mathfrak{N}_B\}$ is the set of orthonormal basis of the representation $V(\varpi_n)$.
Then from \eqref{eq:mat_coeff1} in Lemma \ref{lem:mat_coeff}, we have
\[
\langle v_{\varpi_n}, \tilde g(\t_B)\tilde h\cdot v_{\varpi_n}\rangle_{\varpi_n}=\sum_{r\in\mathfrak{N}_B}
\langle v_{\varpi_n},\tilde g(\t_B)\cdot u_r\rangle_{\varpi_n}\,\langle u_r,\tilde h\cdot v_{\varpi_n}\rangle_{\varpi_n},
\]
which gives the desired formula.
\end{proof}



\subsection{Polynomial solutions of $B$-type}
The polynomial solution of $\tau_k^B(\t_B)$ for each $ w_*\in\mathfrak{W}_B$ is obtained by setting $g(\t_B)=\mathcal{P}_n^B(\t_B)$, i.e.
\begin{align*}
\tau^B_k(\t_B)&=\langle e_1\wedge\cdots\wedge e_k,~\mathcal{P}_n^B(\t_B)n_*\dot w_*\cdot e_{1}\wedge\cdots\wedge e_{k}\rangle,\qquad (1\le k\le n-1),\\
\tau^B_n(\t_B)&=\sqrt{\langle e_1\wedge\cdots\wedge e_n,~\mathcal{P}_n^B(\t_B)n_*\dot w_*\cdot e_{1}\wedge\cdots\wedge e_{n}\rangle},
\end{align*}
where $\mathcal{P}_n^B(\t_B)$ is defined by $\mathcal{P}_{2n+1}(t_1,0,t_3,0,\ldots)$ in \eqref{eq:polyM}.

In order to compute these $\tau$-functions, we first recall that for an index set $\{i_1,i_2,\ldots,i_k\}$, the Schur polynomial corresponding to the index set is given by
\begin{equation}\label{eq:SchurW}
S_{\lambda}(\t_B)=\langle e_1\wedge\cdots\wedge e_{k},\,\mathcal{P}_n^B(\t_B)\cdot e_{i_1}\wedge\cdots\wedge e_{i_k}\rangle\quad\text{with}\quad \lambda=\left(i_k-k,\, i_{k-1}-(k-1),\,\cdots\,, \,i_1-1\right).
\end{equation}
Then we have the following proposition.
\begin{proposition}\label{duality}
Let $\lambda'$ be the transposed Young diagram of $\lambda$. Then the Schur polynomials $S_\lambda(\t_B)$ and $S_{\lambda'}(\t_B)$ are the same, i.e.
\[
S_{\lambda}(\t_B)=S_{\lambda'}(\t_B).
\]
\end{proposition}
To prove this, we start with the following lemma as a special case of Proposition \ref{duality}.
\begin{lemma}\label{lem:dual}
For $\lambda=(k)$, the Young diagram with $k$ horizontal boxes, we have
\[
S_{(k)}(\t_B)=S_{(1^k)}(\t_B),
\]
where the transposed diagram of $(k)$ is $(1^k)$, the diagram with $k$ vertical boxes.
\end{lemma}
\begin{proof}
Recall that
\[
\exp\left(\sum_{i=0}^\infty \kappa^{2i+1}t_{2i+1}\right)=\sum_{k=0}^\infty \kappa^k S_{(k)}(\t_B).
\]
Replacing $\kappa\to -\kappa$ and $\t_B\to-\t_B$, we have
\[
\exp\left(\sum_{i=0}^\infty \kappa^{2i+1}t_{2i+1}\right)=\sum_{k=0}^\infty\kappa^k (-1)^k S_{(k)}(-\t_B)=\sum_{k=0}^\infty \kappa^k S_{(1^k)}(\t_B),
\]
where we have used $(-1)^kS_{(k)}(-\t_B)=S_{(1^k)}(\t_B)$. This proves the lemma.
\end{proof}

\smallskip
Now we prove the proposition.

\smallskip

\emph{Proof of Proposition \ref{duality}}. 
Let $\lambda=(\lambda_1,\lambda_2,\ldots,\lambda_r)$ and denote by $\lambda'=(\lambda'_1,\lambda'_2,\ldots,\lambda_s')$ its transposed diagram.
Recall the Jacobi-Trudi identity of the Schur polynomial,
\[
S_{\lambda}(\t)=\left|\begin{matrix}
h_{\lambda_1}& h_{\lambda_2+1} &\cdots &h_{\lambda_r+r-1}\\
h_{\lambda_1-1}&h_{\lambda_2}&\cdots &h_{\lambda_r+r-2}\\
\vdots &\vdots&\ddots &\vdots\\
h_{\lambda_1-r+1}&h_{\lambda_2-r+2}&\cdots &h_{\lambda_r}
\end{matrix}\right|=\left|\begin{matrix}
e_{\lambda'_1}& e_{\lambda'_2+1} &\cdots &e_{\lambda'_s+s-1}\\
e_{\lambda'_1-1}&e_{\lambda'_2}&\cdots &e_{\lambda'_s+s-2}\\
\vdots &\vdots&\ddots &\vdots\\
e_{\lambda'_1-s+1}&e_{\lambda'_2-s+2}&\cdots &e_{\lambda'_s}
\end{matrix}\right|,
\]
where $h_k=S_{(k)}$ and $e_k=S_{(1^k)}$. Then Lemma \ref{lem:dual}, i.e. $h_k=e_k$ proves the proposition.
$\square$

Now we consider  the Schur polynomial associated with $w\in\mathfrak{W}_B$,
\[
S_{\lambda(w)}(\t_B)=\langle e_1\wedge\cdots\wedge e_n,~ \mathcal{P}_n^B(\t_B)\dot w\cdot e_1\wedge\cdots\wedge e_n\rangle,
\]
where $\lambda(w)=(i_n-n,\ldots,i_2-2,i_1-1)$ with $\{w(1),\ldots,w(n)\}=\{i_1<i_2<\cdots<i_n\}$.
To find $\lambda(w)$, we first note the following lemma.
\begin{lemma}\label{lem:w-r}
For an element $r_k\in\mathfrak{N}_B$, the corresponding Young diagram $\lambda(r_k)$ is expressed in the following Frobenius notation,
\[
\lambda(r_k)=(n-k+1|n-k),
\]
which is the hook diagram $(n-k+2, 1^{n-k})$.
\end{lemma}
\begin{proof}
Since we have
\[
r_k\cdot(1,2,\ldots,n)=(1,\ldots, k-1,\bar{k}, k+1,\ldots,n)\qquad\text{with}\qquad \bar{k}=2n+2-k.
\]
The corresponding Young diagram is $\lambda=(n+2-k, 1^{n-k})$, which is $(n-k+1|n-k)$ in the Frobenius notation.
\end{proof}

Then we have the following proposition on $\lambda(w)$.
\begin{proposition}\label{prop:Frobenius}
In the Frobeneus notation, the Young diagram $\lambda(w)$ can be expressed in the form,
\[
\lambda(w)=(\alpha_1,\alpha_2,\ldots,\alpha_m\,|\,\alpha_1-1,\alpha_2-1,\ldots,\alpha_m-1),
\]
where $\alpha_m\ge 1$ with some constant $m$.
\end{proposition}
\begin{proof}
 Let us take the following example with $n=7$. We consider
 \[
  w\cdot(1,2,3,4,5,6,7)=\dot r_2\dot r_4 \dot r_7(1,2,3,4,5,6,7)=(1,\bar{2},3,\bar{4},5,6,\bar{7})=(1, 14, 3, 12, 5, 6, 9).
 \]
 Then from Lemma \ref{lem:w-r}, we have
 \[
 \lambda(r_2)=(6|5),\qquad \lambda(r_4)=(4|3),\qquad \lambda(r_7)=(1|0).
 \]
 That is, this proposition states that 
 \[
 \lambda(w)=\lambda(r_2r_4r_7)=(6,4,1|5,3,0).
 \]
 This may be understood as follows.
First express $ w$ in the following diagram (Sato's Maya-diagram),
\medskip
\[
\young[15][12][\hskip0.4cm\circle*{8},\wb,\bb,\wb,\bb,\bb,\wb,\wb,\bb,\wb,\wb,\bb,\wb,\bb,\wb]
\]

\medskip
\noindent
which can also be written as 
\[
w\cdot (1,2,\ldots,15)=(1,\overline 2,3,\overline 4,5,6,\overline 7,\overline{\bf 8},9,\overline{10},\overline{11},12,\overline{13},14,\overline{15}),
\]
where $\bar{k}$ for an index $k$ is marked with a white stone at the $k$-th box in the Maya diagram, and other indices without ``\emph{bar}'' are marked with the black stones. In particular, the 8-th ($(n+1)$-th) box from left is always marked with a white stone.  Note that $\overline{k}=16-k$, in general, $\overline{k}=2n+2-k$, and $\overline{n+1}=n+1$.  Then the corresponding Young diagram is given by $\lambda_{n-j+1}=i_j-j$, i.e. in this example, we have
\[
\lambda(w)=(7,6,4,2,2,1,0),
\]
which is
\medskip

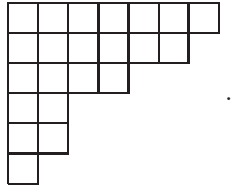
\[
\young[7,6,4,2,2,1][8].
\]

\medskip
\noindent
It is known that the Maya diagram has a direct connection with the Young diagram.
Starting from the bottom left corner, we can take the lattice path along the boundary of the Young diagram in counterclockwise direction and assign each edge with the number from 1 to $2n+1$, so that each horizontal edge is labeled by an index with ``$bar$'' and each vertical edge is labeled by the index without ``$bar$''.
The horizontal and vertical edges in each hook diagram are labeled by $k$ and $\overline{k}=2n+2-k$ for some number $k$. Then it is easy to see that the number of vertical boxes in the hook diagram is given by $n-k+1$, and that the number of the horizontal boxes is $\overline{k}-(n+1)+1=n+2-k$. That is, in the Frobenius 
notation, this hook diagram is expressed by $(n-k+1,n-k)$, which proves the lemma.  Note that the number $m$ in the formula is given by the total number of indices with ``$bar$'' in $w(1,\ldots,n)$, i.e. the length of $r$ in $w$.
In the example, we have $m=3$ with $(\overline{2}, \overline{4},\overline{7})$ ($r=r_2r_4r_7$), and then we have
\[
\lambda(w)=(6,4,1|5,3,0).
\]
\end{proof}

Now we have the Jozefiak-Pragacz type identity  of  the Schur polynomial $S_{\lambda(w)}(\t_B)$ for each $w\in\mathfrak{W}_B$\cite{JP:91,Yo:89}.
\begin{proposition}\label{prop:SchurQ}
Let $\lambda( w)$ be given in the following Frobenius formula,
\[
\lambda(w)=(\alpha_1-1,\alpha_2-1,\ldots,\alpha_s-1\,|\,\alpha_1,\alpha_2,\ldots,\alpha_s).
\]
Then the Schur polynomial can be expressed by
\[
S_{\lambda(w)}(\t_B)=\frac{1}{2^s}Q_{\alpha}(\t_B)^2,
\]
where $Q_{\alpha}(\t_B)$ is the $Q$-Schur polynomial with the Young diagram $\alpha=(\alpha_1,\alpha_2,\ldots,\alpha_s)$.
Note here that if $s$ is odd, then we add $0$ at the end so that $\alpha=(\alpha_1,\ldots,\alpha_s,0)$.
\end{proposition}

\begin{proof} With the Frobenius notation,  the Giambelli formula gives
\[
S_{\lambda(w)}=\left|\begin{matrix}
S_{(\alpha_1-1|\alpha_1)}& S_{(\alpha_1-1|\alpha_2)}&\cdots &S_{(\alpha_1-1|\alpha_s)}\\
S_{(\alpha_2-1|\alpha_1)} &S_{(\alpha_2-1|\alpha_2)}&\cdots &S_{(\alpha_2-1|\alpha_s)}\\
\vdots &\vdots &\ddots & \vdots\\
S_{(\alpha_s-1|\alpha_1)}& S_{(\alpha_s-1|\alpha_2)}&\cdots &S_{(\alpha_s-1|\alpha_s)}
\end{matrix}\right|,
\]
where $(\alpha|\beta)$ is the Young diagram of hook, $(\alpha+1,1^{\beta})$. Then we note that the Schur polynomial $S_{(\alpha-1|\beta)}(t)$ can be expressed by
\[
S_{(\alpha-1|\beta)}=\left|\begin{matrix}
p_1 & p_2 &\cdots & p_{\beta}& p_{\alpha+\beta}\\
1& p_1 & \cdots &p_{\beta-1}&p_{\alpha+\beta-1}\\
0& 1 &\ddots &p_{\beta-2}&p_{\alpha+\beta-2}\\
\vdots &\vdots &\ddots &\vdots &\vdots \\
0&\cdots&\cdots &1 &p_{\alpha}
\end{matrix}\right|.
\]
Expanding in the last column, we have
\[
S_{(\alpha-1|\beta)}=(-1)^{\beta+1}\left(h_{\alpha+\beta}-h_{\alpha+\beta-1}e_1+h_{\alpha+\beta-2}e_2+
\cdots +(-1)^{\beta+1}h_\alpha e_{\beta}\right),
\]
where $h_{k}=p_k=S_{(k)}$ and $e_{k}=S_{(1^k)}$. From Lemma \ref{lem:dual}, $h_k=e_k$, and
noting that $S_{(k)}=Q_k$ (see Appendix B), we have
\[
S_{(\alpha-1|\beta)}=\sum_{i=0}^\beta (-1)^iQ_{\alpha+i}Q_{\beta-i}=\frac{1}{2}\left(Q_{\alpha}Q_{\beta}+Q_{\alpha,\beta}\right),
\]
where we have used the definition of $Q_{\alpha,\beta}$ (see Appendix B).  Then we have
\begin{align*}
S_{\lambda( w)}&=\frac{1}{2^s}\left|\begin{matrix}
Q_{\alpha_1}Q_{\alpha_1} & Q_{\alpha_1,\alpha_2}+Q_{\alpha_1}Q_{\alpha_2} &\cdots &
Q_{\alpha_1,\alpha_s}+Q_{\alpha_1}Q_{\alpha_s} \\
Q_{\alpha_2,\alpha_1}+Q_{\alpha_2}Q_{\alpha_1} & Q_{\alpha_2}Q_{\alpha_2} &\cdots &
Q_{\alpha_2,\alpha_s}+Q_{\alpha_2}Q_{\alpha_s} \\
\vdots &\vdots &\ddots &\vdots\\
Q_{\alpha_s,\alpha_1}+Q_{\alpha_s}Q_{\alpha_1} & Q_{\alpha_s,\alpha_2}+Q_{\alpha_s}Q_{\alpha_2} &\cdots &
Q_{\alpha_s}Q_{\alpha_s} 
\end{matrix}\right|\\[2.0ex]
&=\frac{1}{2^s}\left|\begin{matrix}
0 & Q_{\alpha_1,\alpha_2}&\cdots &Q_{\alpha_1,\alpha_s} &Q_{\alpha_1} \\
Q_{\alpha_2,\alpha_1} & 0 &\cdots &Q_{\alpha_2,\alpha_s}&Q_{\alpha_2} \\
\vdots &\vdots &\ddots &\vdots&\vdots\\
Q_{\alpha_s,\alpha_1} & Q_{\alpha_s,\alpha_2} &\cdots & 0&Q_{\alpha_s}\\
-Q_{\alpha_1}&-Q_{\alpha_2}&\cdots &-Q_{\alpha_s} &1
\end{matrix}\right|\\[2.0ex]
&=\frac{1}{2^s}\left[|Q_{\alpha_i,\alpha_j}|+\sum_{1\le i,j\le s}\Delta_{i,j}Q_{\alpha_i}Q_{\alpha_j}\right],
\end{align*}
where  $\Delta_{i,j}$ is the cofactor of
the skew-symmetric matrix $(Q_{\alpha_i,\alpha_j})$.  Now there are two cases to consider:
\begin{itemize}
\item[(a)] If $s=$ even, then the cofactor $\Delta_{i,j}$ is skew symmetric, $\Delta_{i,j}=-\Delta_{j,i}$. In this case, the summing term in the third line vanishes, and we have
\[
S_{\lambda( w)}=\frac{1}{2^s}|Q_{\alpha_i,\alpha_j}|=\frac{1}{2^s}(Q_{\alpha})^2,
\]
that is, the determinant of the skew symmetric matrix $(Q_{\alpha_i,\alpha_j})$  is the square of the Pfaffian, denoted by $Q_{\alpha}$.
\item[(b)] If $s=$ odd, then the second line can be written by
\[
S_{\lambda( w)}=\frac{1}{2^s}\left|\begin{matrix}
0 & Q_{\alpha_1,\alpha_2}&\cdots &Q_{\alpha_1,\alpha_s} &Q_{\alpha_1} \\
Q_{\alpha_2,\alpha_1} & 0 &\cdots &Q_{\alpha_2,\alpha_s}&Q_{\alpha_2} \\
\vdots &\vdots &\ddots &\vdots&\vdots\\
Q_{\alpha_s,\alpha_1} & Q_{\alpha_s,\alpha_2} &\cdots & 0&
Q_{\alpha_s}\\
-Q_{\alpha_1}&-Q_{\alpha_2}&\cdots &-Q_{\alpha_s} &0
\end{matrix}\right|,
\]
that is, the element 1 at the bottom right corner can be replaced by $0$ (the leading $s\times s$ minor is 0 when $s=$ odd).
Then we have
\[
S_{\lambda( w)}=\frac{1}{2^s}(Q_{\tilde\alpha})^2,
\]
where the Young diagram $\tilde\alpha=(\alpha_1,\ldots,\alpha_s,0)$, i.e. identify $Q_{\alpha_k}=Q_{\alpha_k,0}$
and $-Q_{\alpha_k}=Q_{0,\alpha_k}$. 
\end{itemize}
This completes the proof.
\end{proof}

\subsection{Example: Polynomial solutions of $B_2$-type}
Here we give the polynomial solutions of $(\tau_1^B(\t_B,w_*), \tau_2^B(\t_B,w_*))$ for $w_*=e, s_1^B, s_1^Bs_2^B, s_1^Bs_2^Bs_1^B$ and $ w_0=s_1^Bs_2^Bs_1^Bs_2^B$.
The $\tau_1^B(\t_B)$ is given by
\[
\tau_1^B(\t_B,w_*)=\Delta_1(\t_B)=\langle e_1,\,\mathcal{P}_2^B(\t_B) n_*\dot w_* e_1\rangle,
\]
and for $\tau_2^B(\t_B)$, we use the formula
in Theorem \ref{thm:pfaffian}, i.e.
\[
\tau_2^B(\t_B,w_*)=\sum_{r\in\mathfrak{N}_B}\eta_r(\xi; w_*)\phi_r(\t_B).
\]
Let us first compute the polynomial $\phi_r(\t_B)=\sqrt{\langle e_1\wedge e_2,\mathcal{P}^B_2(\t_B) \dot r\cdot e_1\wedge e_2\rangle}$ for each $r\in\mathfrak{N}_B$.
\begin{itemize}
\item[(a)] For $r=e$, we have
\[
\left(\phi_e(\t_B)\right)^2=\langle e_1\wedge e_2,\mathcal{P}^B_2(\t_B) \cdot e_1\wedge e_2\rangle=1.
\]
\item[(b)] For $r=r_1$, we have
\[
\left(\phi_{r_1}(\t_B)\right)^2=\langle e_1\wedge e_2,\mathcal{P}^B_2(\t_B) \dot{r}_1\cdot e_1\wedge e_2\rangle=S_{\young[3,1][3]}(\t_B)=\frac{1}{2}\left(Q_{\young[2][3]}(\t_B)\right)^2.
\]
\item[(c)] For $r=r_2$, we have
\[
\left(\phi_{r_2}(\t_B)\right)^2=\langle e_1\wedge e_2,\mathcal{P}^B_2(\t_B) \dot{r}_1\cdot e_1\wedge e_2\rangle=S_{\young[2][3]}(\t_B)=\frac{1}{2}\left(Q_{\young[1][3]}(\t_B)\right)^2.
\]
\item[(d)] For $r=r_1r_2$, we have
\[
\left(\phi_{r_1r_2}(\t_B)\right)^2=\langle e_1\wedge e_2,\mathcal{P}^B_2(\t_B) \dot{r}_1\dot{r}_2\cdot e_1\wedge e_2\rangle=S_{\young[3,3][3]}(\t_B)=\frac{1}{4}\left(Q_{\young[2,1][3]}(\t_B)\right)^2.
\]
\end{itemize}
Now we compute $n_*$ and $\eta_r(\xi; w_*)$, and then give $(\tau_1(\t_B,w_*), \tau_2^B(\t_B,w_*))$.
\begin{itemize}
\item[(1)] $ w_*=e$: In this case we have $n_*=y_1^B(\xi_1)y_2^B(\xi_2)y_1^B(\xi_3)y_2^B(\xi_4)$\footnote{
{Let $w_*(1, 2, \dots, 2n+1) = (i_1, i_2, \dots, i_{2n+1})$, then the sign on the $k$-th column of $n_*$ is specified by the number of $i_j (j < k)$'s such that $i_j > i_k$}}, which gives
\[
n_*=\begin{pmatrix}
1 & 0 & 0 & 0 & 0 \\
\xi_1+\xi_3& 1 & 0 & 0 & 0 \\
2\xi_2\xi_3 & 2(\xi_2+\xi_4)& 1 & 0& 0\\
2\xi_2^2\xi_3 &2(\xi_2+\xi_4)^2 & 2(\xi_2+\xi_4) & 1 & 0 \\
2\xi_1\xi_2^2\xi_3 & 2(\xi_3\xi_4^2+\xi_1(\xi_2+\xi_4)^2)&2(\xi_3\xi_4+\xi_1(\xi_2+\xi_4))& \xi_1+\xi_3 & 1
\end{pmatrix}.
\]
The coefficients $\eta_r(\xi;e)=\sqrt{\langle \dot r\cdot e_1\wedge e_2,\,n_*\cdot e_1\wedge e_2\rangle}$ for  $r\in\mathfrak{N}_B$ are given by
\[
\eta_e=1,\quad \eta_{r_1}=\sqrt{2}(\xi_3\xi_4+\xi_1(\xi_2+\xi_4)),\quad \eta_{r_2}=\sqrt{2}(\xi_2+\xi_4),\quad \eta_{r_1r_2}=2\xi_2\xi_3\xi_4.
\]
Then the $\tau$-functions are given by
\[\left\{\begin{array}{lll}
\tau_1^B(\t_B,e)=1+(\xi_1+\xi_3)S_{\young[1][3]}(\t_B)+2\xi_2\xi_3S_{\young[2][3]}(\t_B)+2\xi_2^2\xi_3S_{\young[3][3]}(\t_B)+2\xi_1\xi_2^2\xi_3S_{\young[4][3]}(\t_B),\\[1.0ex]
\tau_2^B(\t_B,e)=1+(\xi_2+\xi_4)Q_{\young[1][3]}(\t_B)+(\xi_3\xi_4+\xi_1(\xi_2+\xi_4))Q_{\young[2][3]}(\t_B)+\xi_2\xi_3\xi_4Q_{\young[2,1][3]}(\t_B).
\end{array}\right.
\]
\item[(2)] $ w_*=s_1^B$: In this case we have $n_*=y_2^B(\xi_1)y_1^B(\xi_2)y_2^B(\xi_3)$, which gives
\[
n_*\dot s_1^B=\begin{pmatrix}
0 & -1 & 0 & 0 & 0 \\
1& -\xi_2 & 0 & 0 & 0 \\
2(\xi_1+\xi_3) & -2\xi_1\xi_3& 1 & 0& 0\\
2(\xi_1+\xi_3)^2 & -2\xi_1^2\xi_2 & 2(\xi_1+\xi_3)& 0 & -1 \\
2\xi_2\xi_3^2& 0&2\xi_2\xi_3& 1 & -\xi_2
\end{pmatrix}.
\]
The coefficients $\eta_r(\xi;s_1^B)=\sqrt{\langle \dot r\cdot e_1\wedge e_2,\,n_*\dot{s}_1^B\cdot e_1\wedge e_2\rangle}$ for  $r\in\mathfrak{N}_B$ are given by
\[
\eta_e=1,\quad \eta_{r_1}=\sqrt{2}\xi_2\xi_3,\quad \eta_{r_2}=\sqrt{2}(\xi_1+\xi_3),\quad \eta_{r_1r_2}=2\xi_1\xi_2\xi_3.
\]
Then the $\tau$-functions are given by
\[\left\{\begin{array}{lll}
\tau_1^B(\t_B,s_1^B)=S_{\young[1][3]}(\t_B)+2(\xi_1+\xi_3)S_{\young[2][3]}(\t_B)+2(\xi_1+\xi_3)^2S_{\young[3][3]}(\t_B)+2\xi_1\xi_2^2S_{\young[4][3]}(\t_B),\\[1.0ex]
\tau_2^B(\t_B,s_1^B)=1+(\xi_1+\xi_3)Q_{\young[1][3]}(\t_B)+\xi_2\xi_3Q_{\young[2][3]}(\t_B)+\xi_1\xi_2\xi_3Q_{\young[2,1][3]}(\t_B).
\end{array}\right.
\]
\item[(3)] $ w_*=s_1^Bs_2^B$: In this case we have $n_*=y_1^B(\xi_1)y_2^B(\xi_2)$, which gives
\[
n_*\dot s_1^B\dot s_2^B=\begin{pmatrix}
0        &   0     &     0 &   -1 & 0 \\
1        & 0     &   0     & -\xi_1& 0 \\
2\xi_2    & 0  & -1 & 0& 0\\
2\xi_2^2   &  0  & -2\xi_2& 0 & -1\\
2\xi_1\xi_2^2 &1     & -2\xi\xi_2   &0& -\xi_1
\end{pmatrix}.
\]
The coefficients $\eta_r(\xi;s_1^Bs_2^B)=\sqrt{\langle \dot r\cdot e_1\wedge e_2,\,n_*\dot{s}_1^B\dot{s}_2^B\cdot e_1\wedge e_2\rangle}$ for  $r\in\mathfrak{N}_B$ are given by
\[
\eta_e=0,\quad \eta_{r_1}=1,\quad \eta_{r_2}=0,\quad \eta_{r_1r_2}=\sqrt{2}\xi_2.
\]
Then the $\tau$-functions are given by
\[\left\{\begin{array}{lll}
\tau_1^B(\t_B,s_1^Bs_2^B)=S_{\young[1][3]}(\t_B)+2\xi_2S_{\young[2][3]}(\t_B)+2\xi_2^2S_{\young[3][3]}(\t_B)+2\xi_1\xi_2^2S_{\young[4][3]}(\t_B),\\[1.0ex]
\tau_2^B(\t_B,s_1^Bs_2^B)=\frac{1}{\sqrt{2}}\left(Q_{\young[2][3]}(\t_B)+\xi_2Q_{\young[2,1][3]}(\t_B)\right).
\end{array}\right.
\]
\item[(4)] $ w_*=s_1^Bs_2^Bs_1^B$: In this case we have $n_*=y_2^B(\xi)$, which gives
\[
n_*\dot s_1^B\dot s_2^B\dot s_1^B=\begin{pmatrix}
0        &   0     &     0 &   0 & 1 \\
0       & -1     &   0     & 0& 0 \\
0    & -2\xi  & -1 & 0& 0\\
0   &  -2\xi^2  & -2\xi&-1 & 0\\
1 &0     & 0  &0& 0
\end{pmatrix}.
\]
The coefficients $\eta_r(\xi;s_1^Bs_2^Bs_1^B)=\sqrt{\langle \dot r\cdot e_1\wedge e_2,\,n_*\dot{s}_1^B\dot{s}_2^B\dot{s}_1^B\cdot e_1\wedge e_2\rangle}$ for  $r\in\mathfrak{N}_B$ are given by
\[
\eta_e=0,\quad \eta_{r_1}=1,\quad \eta_{r_2}=0,\quad \eta_{r_1r_2}=\sqrt{2}\xi.
\]
Then the $\tau$-functions are given by
\[\left\{\begin{array}{lll}
\tau_1^B(\t_B,s_1^Bs_2^Bs_1^B)=S_{\young[4][3]}(\t_B),\\[1.0ex]
\tau_2^B(\t_B,s_1^Bs_2^Bs_1^B)=\frac{1}{\sqrt{2}}\left(Q_{\young[2][3]}(\t_B)+\xi Q_{\young[2,1][3]}(\t_B)\right).
\end{array}\right.
\]
\item[(5)] $ w_*=w_0=s_1^Bs_2^Bs_1^Bs_2^B$ (the longest element): In this case we have $n_*=Id$, which gives
\[
n_*\dot s_1^B\dot s_2^B\dot s_1^B\dot s_2^B=\begin{pmatrix}
0        &   0     &     0 &  0& 1 \\
0       & 0    &   0     & -1& 0 \\
0    & 0 & 1 & 0& 0\\
0  & -1  & 0&0 & 0\\
1 &0    & 0  &0& 0
\end{pmatrix}.
\]
The coefficients $\eta_r(\xi;s_1^Bs_2^Bs_1^Bs_2^B)=\sqrt{\langle \dot r\cdot e_1\wedge e_2,\,n_*\dot{s}_1^B\dot{s}_2^B\dot{s}_1^B\dot{s}_2^B\cdot e_1\wedge e_2\rangle}$ for  $r\in\mathfrak{N}_B$ are given by
\[
\eta_e=0,\quad \eta_{r_1}=0,\quad \eta_{r_2}=0,\quad \eta_{r_1r_2}=1.
\]
Then the $\tau$-functions are given by
\[
\tau_1^B(\t_B,w_0)=S_{\young[4][3]}(\t_B),\qquad\quad \tau_2^B(\t_B,w_0)=Q_{\young[2,1][3]}(\t_B).
\]
\end{itemize}

We also compute $p_4(\tilde D)\Delta_1\circ\Delta_3$ and $p_3(\tilde D)\Delta_1\circ\Delta_2$, which give
the entries $a_{4,1}=a_{5,2}$ and $a_{3,1}=-a_{5,3}$. In particular, $a_{4,1}$ is given by
\[
a_{4,1}=a_{4,1}^0\frac{\Delta_4\Delta_0}{\Delta_3\Delta_1}=\frac{p_{4}(\tilde D)\Delta_1\circ\Delta_3}{\Delta_1\Delta_3}.
\]
which gives
\[
p_4(\tilde D)\Delta_1\circ\Delta_3=a_{4,1}^0\Delta_1.
\]
 If both $a_{4,1}=a_{3,1}=0$, then those $\tau$-functions
give a solution of the tridiagonal KT hierarchy of $B$-type. 

\bigskip

\begin{center}
\begin{tabular}{|c|c|c|c|c|c|cl}
\hline
$ w_*$ & $\tau_1=\Delta_1$&$\tau_2=\sqrt{\Delta_2}$&$\Delta_2=\pm\Delta'_3$ &$p_4(\tilde D)\Delta_1\circ\Delta_3$ & $ p_3(\tilde D)\Delta_1\circ\Delta_2$ \\[1.0ex]
\hline
$e$ &1 & 1 & 1& 0 & 0\\[0.5ex]
\hline
$s_1^B$ & $S_{\young[1][3]}$  &1& 1& 0 &  0 \\[0.5ex]
\hline
$s_2^B$ & $1$&$Q_{\young[1][3]}$&$S_{\young[2][3]} $ & 0 & 0\\[0.5ex]
\hline
$s_1^Bs_2^B$ & $S_{\young[1][3]}$&$Q_{\young[2][3]}$  &$S_{\young[3,1][3]}$& $S_{\young[1][3]} $& $S_{\young[2][3]}$\\[0.5ex]
\hline
$s_2^Bs_1^B$ & $S_{\young[3][3]}$&$Q_{\young[1][3]}$ &$S_{\young[2][3]}$& $0$ & $S_{\young[2][3]}$\\[0.5ex]
\hline
$s_1^Bs_2^Bs_1^B$ & $S_{\young[4][3]}$&$Q_{\young[2][3]}$  &$S_{\young[3,1][3]}$& $S_{\young[4][3]}$& $Q_{\young[2][3]}Q_{\young[2,1][3]}$\\[0.5ex]
\hline
$s_2^Bs_1^Bs_2^B$ & $S_{\young[3][3]}$&$Q_{\young[2,1][3]}$  &$S_{\young[3,3][3]}$& $0$ & $S_{\young[3,3][3]}$\\[0.5ex]
\hline
$s_1^Bs_2^Bs_1^Bs_2^B$ & $S_{\young[4][3]}$&$Q_{\young[2,1][3]}$ &$S_{\young[3,3][3]}$ & 0&0\\[0.5ex]
\hline
\end{tabular}
\end{center}

\bigskip
In the table, $\Delta_3'$ implies that if $\Delta_2=S_{\mu}$, then $\Delta_3=S_{\mu'}$ with the conjugate Young diagram of $\mu$.
Notice that if $ w_*$ is the longest element of a subgroup generated by $s_1^B$ or $s_2^B$ or both,
then these $\tau$-functions give a solution of the tridiagonal KT hierarchy \cite{FH}.
That is, we have the following four cases, $\{e\}, \{s_1^B\}, \{s_2^B\}$ and $\{s_1^B,s_2^B\}$.
Also note that the cases of $s_2^Bs_1^B$ and $s_2^Bs_1^Bs_2^B$ give $a_{4,1}=0$. That is, they are solutions to the 2-banded KT hierarchy, and we have
\[
p_3(\tilde D)\Delta_1\circ\Delta_2=a_{3,1}^0\Delta_2,\qquad\text{and}\qquad a_{3,1}(t)=a_{3,1}^0\frac{1}{\tau^B_1(t)}.
\]



\section{The f-KT hierarchy of $G$-type}\label{sec:GKT}
We embed the Lie algebra of $G_2$-type into a Lie algebra of $B_3$-type, i.e. $\mathfrak{sl}_7(\C)$, by taking the following Chevalley generators,
\begin{align*}
H_1=E_{1,1}-E_{2,2}+2E_{3,3}-2E_{5,5}+E_{6,6}-E_{7,7},\qquad
H_2=E_{2,2}-E_{3,3}+E_{5,5}-E_{6,6},
\end{align*}
the simple root vectors,
\begin{align*}
X_1=E_{1,2}+E_{3,4}+E_{4,5}+E_{6,7},\qquad
X_2=E_{2,3}+E_{5,6},
\end{align*}
and the negative ones,
\begin{align*}
Y_1=E_{2,1}+2E_{4,3}+2E_{5,4}+E_{7,6},\qquad
Y_2=E_{3,2}+E_{6,5}.
\end{align*}
The set of positive roots is given by
\[
\Sigma_G^+=\left\{
\alpha_1,~\alpha_2,~\alpha_1+\alpha_2,~2\alpha_1+\alpha_2,~3\alpha_1+\alpha_2,~3\alpha_1+2\alpha_2\,\right\}.
\]
The negative root vectors are generated by
\[
[Y_\alpha,Y_\beta]=N_{\alpha,\beta}Y_{\alpha+\beta}\qquad \text{if}\qquad \alpha+\beta\in\Sigma_G^+.
\]
Here we take the following vectors,
\begin{align*}
Y_{\alpha\{1,1\}}&=E_{3,1}-2E_{4,2}+2E_{6,4}-E_{7,5},\\
Y_{\alpha\{2,1\}}&=E_{4,1}-E_{5,2}-E_{6,3}+E_{7,4},\\
Y_{\alpha\{3,1\}}&=E_{5,1}-E_{7,3},\\
Y_{\alpha\{3,2\}}&=E_{6,1}+E_{7,2},
\end{align*}
where $\alpha\{i,j\}:=i\alpha_1+j\alpha_2$.  Then the Lax matrix for the f-KT hierarchy of $G$-type is defined by
\begin{align*}
L_G=&\sum_{i=1}^2(a_iH_i+X_i+b_iY_i)+\sum_{i=1}^3b_{i,1}Y_{\alpha\{i,1\}}+b_{3,2}Y_{\alpha\{3,2\}}\\
&=\begin{pmatrix}
a_1  & 1 & 0& 0& 0& 0& 0\\
b_1& a_2-a_1& 1 &0 & 0&0&0\\
b_{1,1} &b_2 & 2a_1-a_2 & 1 &0 &0&0\\
b_{2,1} & -2b_{1,1}&2b_1 & 0 & 1 & 0&0\\
b_{3,1}&-b_{2,1}&0&2b_1&a_2-2a_1&1&0\\
b_{3,2}&0&-b_{2,1}&2b_{1,1}&b_2&a_1-a_2&1\\
0&b_{3,2}&-b_{3,1}&b_{2,1}&-b_{1,1}&b_1&-a_1
\end{pmatrix}.
\end{align*}
The set of eigenvalues of $L_G$ is given by
\[
\{\,0,\, ~ \pm(\kappa_1-\kappa_2),\,~ \pm\kappa_i\,\, (1\le i\le 2)\,\}.
\]
Then we find that the invariants $I_m=\frac{1}{m}\text{tr}(L_G^m)$ give $I_m=0$ for $m=$odd as in the case of $B$-type, and
\begin{align*}
I_2&=\kappa_1^2+\kappa_2^2+(\kappa_1-\kappa_2)^2=2(3a_1^2-3a_1a_2+a_2^2+3b_1+b_2),\\
I_4&=\frac{1}{2}(\kappa_1^4+\kappa_2^4+(\kappa_1-\kappa_2)^4)=\frac{1}{4}I_2^2,\\
I_6&=\frac{1}{3}(\kappa_1^6+\kappa_2^6+(\kappa_1-\kappa_2)^6)=\frac{1}{3}(66a_1^6+2a_2^6+66b_1^3+\cdots+6b_{3,2}).
\end{align*}
Here we omit the full expression of $I_6$, but note that $I_6$ is independent of $I_2$.  Since $I_4=\frac{1}{4}I_2^2$, the $t_3$-flow generated by $I_4$ is given by
\[
\frac{\partial L_G}{\partial t_3}=\frac{1}{2}I_2[\pi_{\mathfrak{b}^+}(\nabla I_2),L_G]=\frac{1}{2}I_2\frac{\partial L_G}{\partial t_1},\qquad
\text{which gives}\quad \left(\frac{\partial}{\partial t_3}-\frac{I_2}{2}\frac{\partial}{\partial t_1}\right)L_G=0.
\]
Then using the shifted coordinates $(t_1+ct_3,t_3)$ with the constant $c=I_2/2$, $L_G$ is invariant under the $t_3$-flow. We then consider the $L_G$ which depends only on $\t_G=(t_1,t_5)$.

The $t_1$-flow of the f-KT hierarchy is then given by
\begin{align*}
\frac{\partial a_i}{\partial t_1}&=b_i,\qquad (i=1,2),\\
\frac{\partial b_1}{\partial t_1}&=-\sum_{j=1}^2(C_{1,j}a_j)b_i+b_{1,1},\qquad \frac{\partial b_2}{\partial t_1}=-\sum_{j=1}^2 (C_{2,j}a_j)b_2-3b_{1,1},\\
\frac{\partial b_{1,1}}{\partial t_1}&=(a_1-a_2)b_{1,1}+b_{2,1},\qquad \frac{\partial b_{2,1}}{\partial t_1}=-a_1b_{2,1}+b_{3,1},\\
\frac{\partial b_{3,1}}{\partial t_1}&=(-3a_1+a_2)b_{3,1}+b_{3,2},\qquad \frac{\partial b_{3,2}}{\partial t_1}=-a_2b_{3,2},
\end{align*}
where the Cartan matrix $(C_{i,j})$ is
\[
C=\begin{pmatrix}
2 & -1\\ -3&2\end{pmatrix}.
\]

As in the case of $A$-type, from $a_{1,1}=a_1$ and $a_{2,2}=a_2-a_1$, we define the $\tau$-functions as
\[
a_{i}=\frac{\partial}{\partial t_1}\ln\Delta_i=\frac{\partial }{\partial t_1}\ln\tau^G_i\qquad (i=1,2),
\]
that is, we have $\tau_1=c_1\Delta_1$ and $\tau_2=c_2\Delta_2$ for some constants $c_1,c_2$.  Then from $a_{3,3}$, we have
\[
a_{3,3}=\frac{p_1(\tilde D)\Delta_3\circ\Delta_2}{\Delta_3\Delta_2}=\frac{\partial}{\partial t_1}\ln\frac{\Delta_3}{\Delta_2}=2a_1-a_2,
\]
which gives
\[
\Delta_3=c'_3(\Delta_1)^2=c_3(\tau^G_1)^2,
\]
where $c_3=c'_3c_1^2$ is a constant.  Recall that $\Delta_3$ is a complete square in the case of $B_3$-type.
Then the $\tau^G_1$ function is a Pfaffian in $G_2$-case. This implies that the solutions of $G_2$-type can be obtained by setting $t_3=0$ in the solution of the $B_3$-type.

Now, from $a_{4,4}=0$, we have
\[
a_{4,4}=\frac{p_1(\tilde D)\Delta_4\circ\Delta_3}{\Delta_4\Delta_3}=\frac{\partial}{\partial t_1}\ln\frac{\Delta_4}{\Delta_3}=0,
\]
which gives
\[
\Delta_4=c'_4\Delta_3=c_4(\tau^G_1)^2.
\]
Similarly, from $a_{5,5}=a_2-2a_1$ and $a_{6,6}=a_1-a_2$, we have
\[
\Delta_5=c_5\tau^G_2,\qquad \Delta_6=c_6\tau^G_1,
\]
for some constants $c_5,c_6$.   

One should note that the element $a_{6,1}=b_{3,2}$ can be expressed as
\[
a_{6,1}=a_{6,1}^0\frac{1}{\Delta_2}=\frac{p_6(\tilde D)\Delta_1\circ\Delta_5}{\Delta_1\Delta_5}.
\]
Since $\Delta_5=c_5\Delta_2$, we have
\[
p_6(\tilde D)\Delta_1\circ\Delta_5=a_{6,1}^0c_5\Delta_1,\quad\text{which gives}\quad p_6(\tilde D)\tau^G_1\circ\tau^G_2=c'_6\tau^G_1,
\]
with some constant $c'_6$.

\subsection{The $\ell$-banded GKT hierarchy}
We have the following banded structure:
\begin{itemize}
\item[(a)] If $b_{3,2}=0$ (i.e. 4-banded GKT), then we have
\[
b_{3,1}=b_{3,1}^0\frac{\Delta_2}{\Delta_1^3}=\frac{p_5(\tilde D)\Delta_1\circ\Delta_4}{\Delta_1\Delta_4}.
\]
Since $\Delta_4\propto (\Delta_1)^2$, we have
\[
p_{5}(\tilde D)\Delta_1\circ\Delta_4\propto \Delta_2.
\]
\item[(b)] If $b_{3,2}=b_{3,1}=0$ (3-banded GKT), then we have
\[
b_{2,1}=b_{2,1}^0\frac{1}{\Delta_1}=\frac{p_4(\tilde D)\Delta_1\circ\Delta_3}{\Delta_1\Delta_3}.
\]
\item[(c)] If $b_{3,2}=b_{3,1}=b_{2,1}=0$ (2-banded GKT), then we have
\[
b_{1,1}=b_{1,1}^0\frac{\Delta_1}{\Delta_2}=\frac{p_3(\tilde D)\Delta_1\circ\Delta_2}{\Delta_1\Delta_2}.
\]
\item[(d)] If $b_{3,2}=b_{3,1}=b_{2,1}=b_{1,1}=0$ (tridiagonal GKT), then we have
\[
b_1=b_1^0\frac{\Delta_2}{\Delta_1^2}=\frac{p_2(\tilde D)\Delta_1\circ\Delta_1}{\Delta_1^2},\qquad\quad
b_2=b_2^0\frac{\Delta_1^3}{\Delta_2^2}=\frac{p_2(\tilde D)\Delta_2\circ\Delta_2}{\Delta_2^2}.
\]
Notice in this case that $b_i$'s can be written in the form,
\[
b_i=b_i^0\prod_{j=1}^2\left(\tau^G_j\right)^{-C_{i,j}},
\]
where $C_{i,j}$ is the Cartan matrix. The $\tau$-functions satisfy
\[
p_2(\tilde D)\tau^G_1\circ\tau^G_1=b_1^0\tau^G_2,\qquad p_2(\tilde D)\tau^G_2\circ\tau^G_2=b_2^0\left(\tau^G_1\right)^3.
\]
\end{itemize}

\subsection{Polynomial solutions of $G$-type}
The Weyl group $\mathfrak{W}_G$ is generated by $\{s_1^G,s_2^G\}$. These elements are given by the Bruhat embedding into $\mathfrak{S}_7$,
\[
s_1^G=s_1^As_6^As_3^As_4^As_3^A\qquad\text{and}\qquad  s_2^G=s_2^As_5^A.
\]
Then $\mathfrak{W}_G$ is given by
\begin{align*}
\mathfrak{W}_G&=\langle\, s_1^G, s_2^G~\Big| ~(s_1^G)^2=(s_2^G)^2=(s_1^Gs_2^G)^6=e\,\rangle,\\
\end{align*}
which has 12 elements.  We also define
\[
y_i^G(\xi)=\exp\left(\xi Y_i\right)\qquad  (i=1,2).
\]
Now according to \eqref{eq:Lusztig1}, for each $w_*^G\in \mathfrak{W}_G$, we define $\hat{w}$ with a reduced expression,
\[
\hat{w}=w_*w_0^{-1}=s_{i_1}^G s_{i_2}^G\cdots s_{i_k}^G,
\]
where $w_0\in\mathfrak{W}_G$ is the longest element. The corresponding unipotent matrix $n_*\in \mathcal{N}^-$ is then defined by
\[
n_*(\xi)=Y_{i_1}(\xi_1) Y_{i_2}(\xi_2) \cdots Y_{i_k}(\xi_k).
\]
Then the minors $\Delta_i(\t_G)$ with $\t_G=(t_1,t_5)$ are computed as
\begin{align*}
\Delta_1(\t_G)&=\langle e_1,\,\mathcal{P}(\t_G)n_*(\xi)w_*\cdot e_1\rangle,\\
\Delta_2(\t_G)&=\langle e_1\wedge e_2,\,\mathcal{P}(\t_G)n_*(\xi)w_*\cdot e_1\wedge e_2\rangle,
\end{align*}
where $\mathcal{P}(\t_G)$ is defined by
\[
\mathcal{P}(\t_G)=\exp\left(t_1X+t_5X^5\right)\qquad \text{with}\quad X:=X_1+X_2.
\]

Now we list the polynomial solutions for $w_*\in\mathfrak{W}_G$, each of which is computed from the formula,
\[
\Delta_i(\t_G)=\langle e_1\wedge\cdots\wedge e_i,\,\mathcal{P}_{7}(\t_G)w_*\cdot e_1\wedge\cdots\wedge e_i\rangle\qquad\text{with}\qquad \t_G=(t_1,t_5).
\]
We list those $\Delta_i$ functions in the table below. Each Young diagram $\lambda$ in the table represents $S_{\lambda}(\t_G)$.

\bigskip

\begin{center}
\begin{tabular}{|c|c|c|c|c|c|}
\hline
$w_*$ & $\tau_1=\Delta_1$&$\tau_2=\Delta_2$&$\Delta_3\propto(\tau_1)^2$ & $\Delta_4=\Delta_3'$ \\[1.0ex]
\hline
$e$ &1 & 1& 1 & 1  \\[1.0ex]
\hline
$s_1^G$ & ${\young[1][3]}$ & $1$&${\young[2][3]}$ &${\young[1,1][3]}$ \\[1.0ex]
\hline
$s_2^G$ & $1$&${\young[1][3]}$ &$1$&$1$  \\[1.0ex]
\hline
$s_2^Gs_1^G$ & ${\young[2][3]}$& ${\young[1][3]}$& ${\young[3,1][3]}$ &${\young[2,1,1][3]}$ \\[1.0ex]
\hline
$s_1^Gs_2^G$ & ${\young[1][3]}$&${\young[3,1][3]}$ &${\young[2][3]}$ &${\young[1,1][3]}$  \\[1.0ex]
\hline
$s_1^Gs_2^Gs_1^G$ & ${\young[4][3]}$&${\young[3,1][3]}$ &$\young[4,3,1][3]$ &$\young[3,2,2,1][3]$ \\[1.0ex]
\hline
$s_2^Gs_1^Gs_2^G$ & ${\young[2][3]}$&${\young[4,2][3]}$ &${\young[3,1][3]}$&$\young[2,1,1][3]$  \\[1.0ex]
\hline
$s_2^Gs_1^Gs_2^Gs_1^G$ & ${\young[5][3]}$&${\young[4,2][3]}$ &${\young[4,4,2][3]}$&$\young[3,3,2,2][3]$ \\[1.0ex]
\hline
$s_1^Gs_2^Gs_1^Gs_2^G$ & ${\young[4][3]}$&${\young[5,4][3]}$  & $\young[4,3,1][3]$&$\young[3,2,2,1][3]$\\[1.0ex]
\hline
$s_1^Gs_2^Gs_1^Gs_2^Gs_1^G$ & ${\young[6][3]}$&${\young[5,4][3]}$ &$\young[4,4,4][3]$ &$\young[3,3,3,3][3]$ \\[1.0ex]
\hline
$s_2^Gs_1^Gs_2^Gs_1^Gs_2^G$ & ${\young[5][3]}$&${\young[5,5][3]}$ &$\young[4,4,4][3]$ & ${\young[3,3,3,3][3]}$ \\[1.0ex]
\hline
$s_1^Gs_2^Gs_1^Gs_2^Gs_1^Gs_2^G$ & ${\young[6][3]}$&${\young[5,5][3]}$ &$\young[4,4,4][3]$ & $\young[3,3,3,3][3]$ \\[1.0ex]
\hline
\end{tabular}
\end{center}

\bigskip

Note here that $\tau_1$ is also the $Q$-Schur polynomial, i.e. $S_{(k)}=p_k=Q_k$, and $\Delta_3, \Delta_4\propto (\tau_1)^2$. Also recall that $\Delta_5\propto \Delta_2'$ and $\Delta_6\propto\Delta_1'$, where $\Delta_k'$ represents the Schur polynomial of the conjugate diagram to that of $\Delta_k$.  Since the coordinates for $G$-type is just $t=(t_1,t_5)$, we have the following relations, 
\[
S_{\young[3][3]}=Q_{\young[2,1][3]},\qquad S_{\young[4][3]}=\frac{1}{2}Q_{\young[3,1][3]},\qquad S_{\young[5][3]}=\frac{1}{2}Q_{\young[3,2][3]},\qquad S_{\young[6][3]}=\frac{1}{2}Q_{\young[3,2,1][3]}.
\]

We also compute $p_6(\tilde D)\Delta_1\circ\Delta_5$, $p_5(\tilde D)\Delta_1\circ\Delta_4$, $p_4(\tilde D)\Delta_1\circ\Delta_3$ and $p_3(\tilde D)\Delta_1\circ\Delta_2$.

\bigskip

\begin{center}
\begin{tabular}{|c|c|c|c|c|c|}
\hline
$w_*$ & $p_6(\tilde D) \Delta_1\circ\Delta_5$ & $p_5(\tilde D)\Delta_1\circ\Delta_4$&$p_4(\tilde D)\Delta_1\circ\Delta_3$ & $ p_3(\tilde D)\Delta_1\circ\Delta_2$ \\[1.0ex]
\hline
$e$ & 0 & 0 & 0&0 \\[1.0ex]
\hline
$s_1^G$ &$0$ &0& 0 &  0 \\[1.0ex]
\hline
$s_2^G$  &$0$&$0$ & 0 & 0 \\[1.0ex]
\hline
$s_1^Gs_2^G$ &  $0$ &${\young[1][3]}$& ${\young[2][3]}$& 1 \\[1.0ex]
\hline
$s_2^Gs_1^G$ & $0$ &$0$ & 0 & ${\young[2][3]}$ \\[1.0ex]
\hline
$s_1^Gs_2^Gs_1^G$  &$0$ &$0$& ${\young[3,2,2,1][3]}$& ${\young[2,2,1][3]}$ \\[1.0ex]
\hline
$s_2^Gs_1^Gs_2^G$ & ${\young[1,1][3]}$&$\young[1][3]$ & $\young[2][3]$ & ${\young[3,1,1][3]}$ \\[1.0ex]
\hline
$s_1^Gs_2^Gs_1^Gs_2^G$ & ${\young[1,1,1,1,1][3]}$&$\young[2,2,2,2,2][3]$& $\young[3,3,2][3]$ & ${\young[3,3,2][3]}$ \\[1.0ex]
\hline
$s_2^Gs_1^Gs_2^Gs_1^G$ &0&0& ${\young[3,2,2,1][3]}$ & ${\young[4,3,3][3]}$ \\[1.0ex]
\hline
$s_1^Gs_2^Gs_1^Gs_2^Gs_1^G$ & $0$ &$0$& 0 & ${\young[4,4,4][3]}$ \\[1.0ex]
\hline
$s_2^Gs_1^Gs_2^Gs_1^Gs_2^G$ & $0$ & $0$&${\young[3,3,3,3][3]}$ & ${\young[4,4,4][3]}$ \\[1.0ex]
\hline
$s_1^Gs_2^Gs_1^Gs_2^Gs_1^Gs_2^G$ & $0$ & $0$&0 & $0$ \\[1.0ex]
\hline
\end{tabular}
\end{center}

\bigskip

Here note that we have the following relations,
\[
S_{\young[2][3]}=\frac{1}{2}(Q_{\young[1][2]})^2,\qquad \qquad S_{\young[4,4,4][3]}=\frac{1}{2}(Q_{\young[3,2,1][3]})^2=\frac{1}{2}(Q_{\young[6][3]})^2.
\]
Note that if $w_*$ is the longest element of a subgroup generated by $s_1^G$ or $s_2^G$ or both,
then these $\tau$-functions give a solution of the tridiagonal KT hierarchy \cite{FH}.
That is, we have the following four cases, $\{e\}, \{s_1^G\}, \{s_2^G\}$ and $\{s_1^G,s_2^G\}$.

For the case with $w_*=s_1^Gs_2^G$, these $\tau$-functions give a solution of the 4-banded KT hierarchy ($b_{3,2}=0$), and we have
\[
p_5(\tilde D)\Delta_1\circ\Delta_4=b_{3,1}^0\Delta_2.
\]

For the cases with $w_*=s_1^Gs_2^Gs_1^G, s_2^Gs_1^Gs_2^Gs_1^G$ and $s_2^Gs_1^Gs_2^Gs_1^Gs_2^G$, we have a solution of the 3-banded KT hierarchy ($b_{3,2}=b_{3,1}=0$), and
\[
p_4(\tilde D)\Delta_1\circ\Delta_3=b_{2,1}^0\Delta_3.
\]

And for the cases of $w_*=s_2^Gs_1^G$ and $s_1^Gs_2^Gs_1^Gs_2^Gs_1^G$, we have a solution of the 2-banded KT-hierarchy ($b_{3,2}=b_{3,1}=b_{2,1}=0$), and 
\[
p_3(\tilde D)\Delta_1\circ\Delta_2=b_{1,1}^0(\Delta_1)^2.
\]


\appendix
\section{The full Kostant-Toda hierarchy}
Here we derive the formula \eqref{aij} in Proposition \ref{prop:aij}.
It is shown in \cite{KY} that the f-KT equation can be obtained as a root reduction of
the Lax equation for the general matrix $\tilde L$,
\[
\frac{\partial \tilde L}{\partial t}=[\tilde P,\tilde L],\qquad \text{with}\quad 
\tilde P=\frac{1}{2}\left((\tilde L)_{>0}-(\tilde L)_{<0}\right),
\]
where $(\tilde L)_{>0}$ and $(\tilde L)_{<0}$ represent the upper and lower triangular parts of $\tilde L$.
Define the eigenmatrices $\tilde \Phi$ and $\tilde\Psi$ which satisfy
\[
\tilde L\tilde \Phi=\tilde \Phi\Lambda,\qquad \tilde \Psi \tilde{L} = \Lambda\tilde\Psi\quad\text{with}\quad \Lambda=\text{diag}(\kappa_1,\ldots,\kappa_M),
\]
with the orthogonality relations,
\[
\tilde\Phi\tilde\Psi=\tilde\Psi\tilde\Phi=I=\text{identity matrix}.
\]
Writing $\tilde\Phi$ and $\tilde\Psi$ in the form,
\[
\tilde\Phi=\begin{pmatrix}
\tilde\phi_1(\kappa_1) &\cdots &\tilde\phi_1(\kappa_M)\\
\vdots &\ddots &\vdots\\
\tilde\phi_M(\kappa_1)&\cdots&\tilde\phi_M(\kappa_M)
\end{pmatrix},
\qquad
\tilde\Psi=\begin{pmatrix}
\tilde\psi_1(\kappa_1) &\cdots &\tilde\psi_M(\kappa_1)\\
\vdots &\ddots &\vdots\\
\tilde\psi_1(\kappa_M)&\cdots&\tilde\psi_M(\kappa_M)
\end{pmatrix}.
\]
Then the element $\tilde a_{i,j}$ in $\tilde L$ can be expressed as
\begin{equation}\label{Taij}
\tilde a_{i,j}=\langle \kappa\tilde\phi_i\tilde\psi_j\rangle:=\sum_{k=1}^n\kappa_k\tilde\phi_i(\kappa_k)\tilde\psi_j(\kappa_k).
\end{equation}
The elements in the eigenmatrices are then given by
\begin{align*}
\tilde\phi_k(t,\kappa)&=\frac{e^{\frac{1}{2}\theta(t,\kappa)}}{\sqrt{D_k(t)D_{k-1}(t)}}\left|
\begin{matrix}
c_{1,1}(t) &\cdots &c_{1,k-1}(t)&\tilde\phi_1^0(\kappa)\\
\vdots &\ddots &\vdots &\vdots \\
c_{k,1}(t) &\cdots &c_{k,k-1}(t) &\tilde\phi_k^0(\kappa)
\end{matrix}\right|,\\[2.0ex]
\tilde\psi_k(t,\kappa)&=\frac{e^{\frac{1}{2}\theta(t,\kappa)}}{\sqrt{D_k(t)D_{k-1}(t)}}\left|
\begin{matrix}
c_{1,1}(t) &\cdots &c_{1,k-1}(t)\\
\vdots &\ddots &\vdots\\
c_{k,1}(t) &\cdots &c_{k,k-1}(t)\\
\tilde\psi_1^0(\kappa)&\cdots &\tilde\psi_k^0(\kappa)
\end{matrix}\right|,
\end{align*}
where  $\tilde\phi_k^0(\kappa)=\tilde\phi_k(0,\kappa),~ \tilde\psi_k^0(\kappa)=\tilde\psi(0,\kappa)$, and $D_k(t)=\text{det}(c_{a,b})_{1\le a,b\le k}$ with $c_{a,b}(t)$ given by
\[
c_{a,b}(t)=\langle\tilde\phi_a^0\tilde\psi_b^0\cdot e^{\theta}\rangle.
\]
Here $\theta(t,\kappa)=\kappa t$ which can be extended to $\theta(\t,\kappa)=\sum_{m=1}^\infty\kappa^m t_m$ for the hierarchy.

For the f-KT equation (non-normalized form), the eigenmatrix $\Phi$ can be written in the form,
\[
\tilde\phi_k(t,\kappa)=\frac{\tilde\phi_1^0(\kappa)e^{\frac{1}{2}\theta(t,\kappa)}}{\sqrt{D_k(t)D_{k-1}(t)}}
\left|
\begin{matrix}
\hat c_{1,1}(t) &\cdots &\hat c_{1,k-1}(t)&1\\
\hat c_{2,1}(t)&\cdots &\hat c_{2,k-1}(t)&\kappa\\
\vdots &\ddots &\vdots &\vdots \\
\hat c_{k,1}(t) &\cdots &\hat c_{k,k-1}(t) &\kappa^{k-1}
\end{matrix}\right|\qquad \text{with}\quad \hat c_{i,j}=\langle \kappa^{i-1}\tilde\phi_1^0\tilde\psi_j^0e^{\theta}\rangle.
\]
Here we have assumed the initial matrix $L^0$ to be 
\[
L^0=\begin{pmatrix}
\tilde a^0_{1,1} & 1 & 0 &\cdots &0\\
\tilde a^0_{2,1}&\tilde  a^0_{2,2} &1&\cdots &0\\
\vdots & \vdots &\ddots & \ddots &\vdots\\
\tilde a^0_{M-1,1}&\tilde a^0_{M-1,2}&\cdots&\ddots&1\\
\tilde a^0_{M,1} &\tilde a^0_{M,2}&\cdots &\cdots & \tilde a^0_{M,M}
\end{pmatrix},
\]
that is, $a^0_{i,i+1}=1$ for all $i=1,\ldots,M-1$.  Then the determinant $D_k(t)$ can be expressed by
\[
D_k(t)=\text{det}\left(\langle \kappa^{a-1}\tilde\phi_a^0\tilde\psi_b e^{\theta}\rangle\right)_{1\le a,b\le k}.
\]
Since $\theta=\kappa t$, we note that $D_k(t)$ can be given by the Wronskian,
\[
D_k(t)=\text{Wr}(f_1(t), f_2(t),\ldots,f_k(t))\qquad\text{with}\quad f_j(t)=\langle \tilde\phi^0_1\tilde\psi_j^0 e^\theta\rangle.
\]
That is, $D_k(t)$ is nothing but the $\tau_k$-function in \eqref{eq:tau1}, and we now use
the multiple time $\t=(t_1,t_2,\ldots)$ to study the f-KT hierarchy.

Now we note that $\tilde\phi_k(\t,\kappa)$ can be expresses as
\begin{align*}
\tilde \phi_k(\t,\kappa)&=\frac{\kappa^{k-1}\tilde\phi_1^0(\kappa) e^{\frac{1}{2}\theta(\t,\kappa)}}{\tau_k(\t)\tau_{k-1}(\t)}\left|
\begin{matrix}
f_1(\t) &\cdots &f_{k-1}(\t) &\frac{1}{\kappa^{k-1}}\\
\partial f_1(\t)&\cdots &\partial f_{k-1}(\t)&\frac{1}{\kappa^{k-2}}\\
\vdots&\ddots &\vdots &\vdots\\
\partial^{k-1}f_1(\t) &\cdots &\partial^{k-1}f_{k-1}(\t) &1
\end{matrix}\right|\\[2.0ex]
&=\frac{\kappa^{k-1}\tilde\phi_1^0(\kappa) e^{\frac{1}{2}\theta(\t,\kappa)}}{\tau_k(\t)\tau_{k-1}(\t)}\left|
\begin{matrix}
f_1(\t)-\frac{1}{\kappa}\partial f_1(\t) &\cdots &f_{k-1}(\t)-\frac{1}{\kappa}\partial f_{k-1}(\t) \\
\partial f_1(\t)-\frac{1}{\kappa}\partial^2 f_1(\t)&\cdots &\partial f_{k-1}(\t)-\frac{1}{\kappa}\partial^2f_{k-1}(\t)\\
\vdots&\ddots &\vdots \\
\partial^{k-2}f_1(\t)-\frac{1}{\kappa}\partial^{k-1}f_1(\t)&\cdots & \partial^{k-2}f_{k-1}(\t)-\frac{1}{\kappa}\partial^{k-1}f_{k-1}(\t)
\end{matrix}\right|,
\end{align*}
where $\partial^j f_i(\t)=\frac{\partial^j}{\partial t_1^j}f_i(\t)$.
Noting $\exp(-\sum_i\frac{k^i}{i\kappa^i})=1-\frac{k}{\kappa}$, we find
\[
f_i(\t)-\frac{1}{\kappa}\partial f_i(\t)=f\left(\t-\left[\frac{1}{\kappa}\right]\right),
\]
where 
\[
\t-\left[\frac{1}{\kappa}\right]=\left(t_1-\frac{1}{\kappa},~ t_2-\frac{1}{2\kappa^2},~ t_3-\frac{1}{3\kappa^3},~\ldots,\right).
\]
Then we have
\begin{align*}
\tilde\phi_k(\t,\kappa)&=\frac{\kappa^{k-1}\tilde\phi_1^0(\kappa) e^{\frac{1}{2}\theta(\t,\kappa)}}{\tau_k(\t)\tau_{k-1}(\t)}\tau_{k-1}\left(\t-\left[\frac{1}{\kappa}\right]\right)\\
&=\frac{\kappa^{k-1}\tilde\phi_1^0(\kappa) e^{\frac{1}{2}\theta(\t,\kappa)}}{\tau_k(\t)\tau_{k-1}(\t)}
e^{-\sum \frac{1}{n\kappa^n}\partial_n}\tau_{k-1}(\t)\\
&=\frac{\kappa^{k-1}\tilde\phi_1^0(\kappa) e^{\frac{1}{2}\theta(\t,\kappa)}}{\tau_k(\t)\tau_{k-1}(\t)}
\sum_{j=0}^{k-1}\frac{1}{\kappa^n}p_j(-\tilde\partial)\tau_{k-1}(\t).
\end{align*}
Here we have used the relation $p_j(-\tilde\partial)\tau_{k-1}(\t)=0$ for $j\ge k$. 

We also write the eigenfunction $\tilde\psi_k$ in the form,
\[
\tilde\psi_j(\t,\kappa)=\frac{e^{\frac{1}{2}\theta(\t,\kappa)}}{\sqrt{\tau_j(\t)\tau_{j-1}(\t)}}
\left|\begin{matrix}
f_1(\t) &\cdots & f_j(\t)\\
\partial_1 f_1(\t)&\cdots& \partial_1 f_j(\t)\\
\vdots & \ddots &\vdots \\
\partial_1^{j-2}f_1(\t)&\cdots &\partial_1^{j-2}f_j(\t)\\
\tilde \psi_1^0(\kappa)&\cdots &\tilde\psi_j^0(\kappa)
\end{matrix}\right|.
\]
In the computation of $\langle \kappa\tilde\phi_i\tilde\psi_j\rangle$, we have
\[
\langle\kappa\tilde\phi_i\tilde\psi_k^0e^{\frac{1}{2}\theta}\rangle=\frac{1}{\sqrt{\tau_i(\t)\tau_{i-1}(\t)}}\sum_{n=0}p_n(-\tilde\partial)\tau_{i-1}(\t)\langle\kappa^{i-n}\tilde\phi_i^0\tilde\psi_k^0e^{\theta}\rangle,
\]
where $\langle\kappa^{i-n}\tilde\phi_i^0\tilde\psi_k^0e^{\theta}\rangle=\partial_1^{i-n}f_k(\t)$.
We also note that
\[
\left|\begin{matrix}
f_1(\t) &\cdots & f_j(\t)\\
\partial_1 f_1(\t)&\cdots& \partial_1 f_j(\t)\\
\vdots & \ddots &\vdots \\
\partial_1^{j-2}f_1(\t)&\cdots &\partial_1^{j-2}f_j(\t)\\
\partial_1^{i-n}f_1(\t)&\cdots &\partial_1^{i-n}f_j(\t)
\end{matrix}\right|=p_{i-M-j+1}(\tilde\partial)\tau_j(\t).
\]
Thus, we have
\[
\tilde a_{i,j}=\langle \kappa\tilde\phi_i\tilde\psi_j\rangle=\frac{1}{\sqrt{\tau_i(\t)\tau_j(\t)\tau_{i-1}(\t)\tau_{j-1}(\t)}}\sum_{n=0}^{i-1}p_n(-\tilde\partial)\tau_{i-1}(\t)\times p_{i-M-j+1}(\tilde\partial)\tau_j(\t).
\]
Using the identity,
\[
p_n(\tilde D)\tau_\alpha\circ\tau_\beta=\sum_{i+j=n}p_i(\tilde\partial)\tau_\alpha\times p_j(-\tilde\partial)\tau_{\beta},
\]
the element $\tilde a_{i,j}$ of the Lax matrix is expressed by
\begin{equation}\label{Taij-tau}
\tilde a_{i,j}=\frac{1}{\sqrt{\tau_i(\t)\tau_j(\t)\tau_{i-1}(\t)\tau_{j-1}(\t)}}p_{i-j+1}(\tilde D)\tau_j\circ\tau_{i-1}.
\end{equation}

The Lax matrix $L$ of the f-KT equation is given by
\[
L=H\tilde L H^{-1}\qquad\text{with}\qquad H=\text{diag}(1,\tilde a_{1,2}, \tilde a_{1,2}a_{2,3},\ldots,a_{1,2}\cdots a_{M-1,M}).
\]
Each entry $a_{i,j}$ is then expressed by 
\[
a_{i,j}=\tilde a_{i,j} \tilde a_{j,j+1}\cdots \tilde a_{i-1,i}\qquad\text{for}\quad i>j.
\]
Using \eqref{Taij-tau}, we first note
\[
\tilde a_{i,i+1}=\frac{\sqrt{\tau_{i+1}(\t)\tau_{i-1}(\t)}}{\tau_i(\t)}.
\]
Then we obtain
\begin{equation}\label{Aij}
 a_{i,j}(\t)=\frac{p_{i-j+1}(\tilde D)\tau_j\circ\tau_{i-1}}{\tau_j\tau_{i-1}}\qquad\text{for}\quad 1\le j\le i\le M.
\end{equation}

\section{The $Q$-Schur functions}
The $Q$-Schur functions are defined as follows. First define the elementary Schur functions on the \emph{odd} times, $\t_B=(t_1,t_3,t_5,\ldots)$, by
\[
e^{\hat\theta(\t_B,k)}=\sum_{r=0}^\infty q_r(\t_B)k^r\qquad\text{with}\quad \hat\theta(\t_B,k)=\sum_{n=\text{odd}}t_nk^n.
\]
Note here that we take $t_n$ in the exponent $\hat\theta(\t_B,k)$ instead of $2t_n$ which is normally used in \cite{Yo:89, JP:91}, and
then the function $q_r(\t_B)$ is given by
\[
q_r(\t_B)=p_r(t_1,0,t_3,0,t_5,0,\ldots).
\]
With those $q_r(\t_B)$ functions, define 
\[
Q_{j,i}(\t_B)=q_j(\t_B)q_i(\t_B)+2\sum_{l=1}^i(-1)^lq_{j+l}(\t_B)q_{i-l}(\t_B),\quad\text{with}\quad Q_{k,0}=q_k.
\]
Note that this function is skew-symmetric $Q_{i,j}=-Q_{j,i}$ (c.f. \cite{JP:91}).
Then the $Q$-Schur function for a strict partition, $\lambda=(\lambda_1>\lambda_2>\cdots>\lambda_{2n}\ge 0)$,
is defined by the Pfaffian,
\begin{align*}
Q_{\lambda}(\t_B):&=\text{Pf}(Q_{\lambda_j,\lambda_i}(\t_B))_{1\le i<j\le 2n}\\[1.0ex]
&=\sum \text{sgn}(i_1,i_2,\ldots,i_{2n})\,Q_{\lambda_{i_1},\lambda_{i_2}}(\t_B)\,Q_{\lambda_{i_3},\lambda_{i_4}}(\t_B)\,\cdots\,
Q_{\lambda_{i_{2n-1}},\lambda_{i_{2n}}}(\t_B).
\end{align*}
For example, if $\lambda=(3,2,1,0)$, then
\[
Q_{3,2,1,0}=Q_{3,2}Q_{1,0}-Q_{3,1}Q_{2,0}+Q_{3,0}Q_{2,1}.
\]
With the Young diagrams, this is also expressed as
\[
Q_{\young[3,2,1][3]}=Q_{\young[3,2][3]}Q_{\young[1][3]}-Q_{\young[3,1][3]}Q_{\young[2][3]}+Q_{\young[3][3]}Q_{\young[2,1][3]}.
\]

\raggedright

\bibliographystyle{amsalpha}
\bibliography{bibliography}
\label{sec:biblio}

\end{document}